\documentclass[11pt]{article}
 \usepackage{amsfonts}
 \usepackage{mathrsfs}
 \usepackage{bbm}
 \usepackage{amsthm}
\usepackage{amsmath,amssymb}
\def\MatrixFont{\bf}
\def\VectorFont{\bf}

\newcommand{\mA}{{\MatrixFont A}}
\newcommand{\mB}{{\MatrixFont B}}
\newcommand{\mC}{{\MatrixFont C}}

\newcommand{\mE}{{\MatrixFont E}}
\newcommand{\mF}{{\MatrixFont F}}

\newcommand{\mI}{{\MatrixFont I}}
\newcommand{\mJ}{{\MatrixFont J}}
\newcommand{\mK}{{\MatrixFont K}}

\newcommand{\mM}{{\MatrixFont M}}
\newcommand{\mN}{{\MatrixFont N}}

\newcommand{\mP}{{\MatrixFont P}}
\newcommand{\mQ}{{\MatrixFont Q}}
\newcommand{\mR}{{\MatrixFont R}}

\newcommand{\mT}{{\MatrixFont T}}

\newcommand{\mV}{{\MatrixFont V}}
\newcommand{\mW}{{\MatrixFont W}}
\newcommand{\mX}{{\MatrixFont X}}
\newcommand{\mY}{{\MatrixFont Y}}
\newcommand{\mZ}{{\MatrixFont Z}}

\newcommand{\ve}{{\VectorFont e}}

\newcommand{\vu}{{\VectorFont u}}
\newcommand{\vv}{{\VectorFont v}}
\newcommand{\vw}{{\VectorFont w}}
\newcommand{\vx}{{\VectorFont x}}
\newcommand{\vy}{{\VectorFont y}}
\newcommand{\vz}{{\VectorFont z}}

\def\diag{\qopname\relax o{diag}}

\newtheorem{theorem}{Theorem}[section]
\newtheorem{lemma}[theorem]{Lemma}

\theoremstyle{definition}

\theoremstyle{remark}
\newtheorem{remark}[theorem]{Remark}

\theoremstyle{algorithm}
\newtheorem{algorithm}[theorem]{Algorithm}

\theoremstyle{corollary}

\theoremstyle{example}

\usepackage{amssymb}
\usepackage{epsfig}
\usepackage{times}
\title{Set-Membership Information Fusion for Multisensor Nonlinear Dynamic Systems}

\author{Zhiguo Wang,  Xiaojing Shen, and Yunmin Zhu  \thanks{This work was supported in part by  the open research funds of BACC-STAFDL of China
under Grant No. 2015afdl010,  the special funds of NEDD of China under Grant No. 201314, the NSFC No. 61673282 and the PCSIRT15R53. Zhiguo Wang, Xiaojing Shen (corresponding author), and Yunmin Zhu  are  with Department of Mathematics, Sichuan
University, Chengdu, Sichuan 610064, China. E-mail: wangzg315@126.com, shenxj@scu.edu.cn,  ymzhu@scu.edu.cn.}}

\begin{document}
 \maketitle
\begin{abstract}
The set-membership information fusion problem is investigated for general multisensor nonlinear dynamic systems. Compared with linear dynamic systems and point estimation fusion in mean squared error  sense, it is a more challenging nonconvex optimization problem. Usually, to solve this problem, people try to find an efficient or heuristic fusion algorithm.   It is no doubt that an analytical fusion formula should be much significant for rasing accuracy and reducing computational burden. However, since it is a more complicated than the convex quadratic optimization problem for linear point estimation fusion, it is not easy to get the analytical fusion formula. In order to overcome the difficulty of this problem, two popular fusion architectures are considered: centralized and distributed set-membership information fusion. Firstly, both of them can be converted into a semidefinite programming  problem which can be efficiently computed, respectively. Secondly,  their analytical solutions can be derived surprisingly by using decoupling technique. It is very interesting that they are quite similar in form to the classic information filter. In the two analytical fusion formulae, the information of each sensor can be clearly characterized, and the knowledge of the correlation among measurement noises across sensors are not required. Finally,  multi-algorithm fusion is used to minimize the size of the state bounding ellipsoid by complementary advantages of multiple parallel algorithms. A typical numerical example in target tracking demonstrates the effectiveness of the centralized, distributed, and multi-algorithm set-membership fusion algorithms. In particular, it shows that multi-algorithm fusion performs better than the centralized and distributed fusion.
\end{abstract}

\noindent{\bf keywords:} Nonlinear dynamic systems, multisensor fusion, target tracking, unknown but bounded noise, set-membership filter

\section{Introduction}\label{sec_1}
In recent years, the multisensor estimation fusion or data fusion has received significant attention for target tracking, artificial intelligence, sensor networks and big data (see \cite{Goodman-Mahler-Nguyen97,Zhu-Li06,Zhu-Zhou-Shen-Song-Luo12,Zheng15}), since many practical problems involve information or data from multiple sources. The problem of multisensor estimation fusion is that how to optimally fuse sensor data from multiple sensors to provide more useful and accurate information for the purpose of estimating an unknown process state \cite{Li-Zhu-Wang-Han03}. Currently, the estimation fusion technology has rapidly evolved from a loosely related techniques to an emerging real engineering discipline with standardized terminology \cite{Hall-Llinas97}.

Generally speaking, there are two traditional architectures for estimation fusion, namely, centralized fusion structure and distributed fusion structure. The centralized architecture is sending the raw data of each sensor to the fusion center, theoretically, which is nothing but an estimation problem with distributed data. Moreover, the centralized fusion approach can usually reach optimal linear estimation in mean squared error (MSE) sense \cite{Hall-Llinas97}\footnote{For nonlinear estimation, the centralized fusion cannot guarantee in general to  reach the optimal estimation.}.
However, the distributed architecture is propagating the estimation of each sensor to the fusion center, which decreases computational burden in the fusion center, but it may not get the optimal linear estimation in MSE sensse. Due to its important practical significance, distributed estimation fusion has been studied extensively, see \cite{Li-Zhu-Wang-Han03}, \cite{Liggins-Chong-Kadar-Alford-Vannicola-Thomopoulos97}, \cite{Fang-Li09}, \cite{Kar-Varshney13}, \cite{Vempaty-Han-Varshney14}.

For multisensor point estimation fusion in probabilistic setting, many results have been obtained (see, e.g., books \cite{BarShalom-Li95}, \cite{Varshney97}, \cite{Zhu03}).  \cite{Li-Zhu-Wang-Han03} provides the optimal linear estimation fusion method for a unified linear model. \cite{Hashemipour-Roy-Laub88} proves that the distributed fusion algorithm is equivalent to the optimal centralized Kalman filtering in the case of cross-uncorrelated sensor noises,  and the one for the case of cross-correlated sensor noises is proposed in \cite{Song-Zhu-Zhou-You07}. When there exists the limitation of communication bandwidth between a fusion center and sensors, \cite{Zhu-Song-Zhou-You05} achieves a constrained optimal estimation at the fusion center. In addition, \cite{Duan-Li11} proposes lossless linear transformation of the raw measurements of each sensor for distributed estimation fusion. Most existing information fusion algorithms are based on the sequential estimation techniques such as Kalman filter, information filter and the weighted least-squares methods \cite{BarShalom-Li-Kirubarajan01}, which need to know the accurate statistical knowledge of the process and measurement noises.

Since the limitation of human and material resources in real life, we cannot obtain the exact statistical characteristics of noise, which may lead to poor performance for the state estimation (see \cite{Theodor-Shaked-deSouza94}, \cite{Zhu12}). Especially for the nonlinear target tracking systems, it is more sensitive to the precise distribution information of noise. In many engineering applications, it is easier to obtain the upper bound and lower bound of a unknown noise \cite{Jaulin-Kieffer-Didrit-Walter01}. In the unknown but bounded setting, the earliest work about the set-membership filter is proposed by \cite{Schweppe68} at the end of 1960s, and it is later developed by \cite{Levinbook-Wong08} and \cite{ElGhaoui-Calafiore01}. These robust filters are derived through set-membership estimate, usually a bounding ellipsoid of containing the true state. Moreover, the set-membership filter for nonlinear dynamic system has also been  investigated by \cite{Scholte-Campbell03}, \cite{Becis-Aubry-Ramdani12}, \cite{Wang-Shen-Zhu-Pan16} and references therein.

For multisensor set-membership fusion in bounded setting, \cite{Wang-Li12} proposes a relaxed Chebyshev center covariance
intersection (CI) algorithm to fuse the local estimates, geometrically, which is the center of the minimum radius ball enclosing the intersection of estimated ellipsoids of each sensor. In order to account for the inconsistency problem of the local estimates, \cite{Uhlmann03} proposes a covariance union method (CU) and it is more conservative than CI fusion. However, the judgment and calculation about correlation may be difficult. Since the set-membership filter only needs to know the bound of the noises, rather than the statistical properties of noises, it does not require to judge the correlation between each sensor, which inspires us to consider set-membership information fusion. When the dynamic system is linear dynamic systems, \cite{Shen-Zhu-Song-Luo11} proposes some algorithms of multisensor set-membership information fusion to minimize Euclidean estimation error of the state vector. However, for  nonlinear dynamic systems, the multisensor set-membership information fusion has not received enough research attention. These facts motivate us to further research the more challenging set-membership fusion problem for nonlinear dynamic systems.


In this paper, two popular fusion architectures are considered: centralized and distributed set-membership information fusion. Firstly, both of them can be converted into a semidefinite programming (SDP) problem which can be efficiently computed, respectively. Secondly,  their analytical solutions can be derived surprisingly by using decoupling technique. It is very interesting that they are quite similar in form to the classic information filter in MSE sense \cite{BarShalom-Li-Kirubarajan01}. In the two analytical fusion formulae, the information of each sensor can be clearly characterized, and the knowledge of the correlation among measurement noises across sensors are not required. Finally,  multi-algorithm fusion is used to minimize the size of the state bounding ellipsoid by complementary advantages of multiple parallel algorithms. A typical numerical example in target tracking demonstrates the effectiveness of the centralized, distributed, and multi-algorithm set-membership fusion algorithms. In particular, it shows that multi-algorithm fusion performs better than both the centralized and distributed fusion.

The rest of the paper is organized as follows. Section \ref{sec_2} introduces the problem formulation for the centralized fusion and the distributed fusion. In Section \ref{sec_3}, the centralized set-membership information fusion algorithm is derived by $\mathcal {S}$-procedure, Schur complement and decoupling technique. Section \ref{sec_4} provides the distributed set-membership information fusion algorithm. A typical example in target tracking is presented in Section \ref{sec_5}, while conclusion is drawn in Section \ref{sec_6}.

\section{Preliminaries}\label{sec_2}
\subsection{Problem Formulation for Centralized Fusion}\label{sec_2_1}
Consider the $L$-sensor \emph{centralized} nonlinear dynamic system with unknown but bounded noises as follows:
\begin{eqnarray}%
\label{Eqpre_1}\vx_{k+1}&=&f_k(\vx_k)+\vw_k,\\
\label{Eqpre_2}\vy_k^i&=&h_k^i(\vx_k)+\vv_k^i,~i=1,\ldots, L,
\end{eqnarray}
where $\vx_k\in \mathcal {R}^n$ is the state of system at time $k$, $\vy_k^i\in \mathcal {R}^{m}$ is the measurement at the $i$th sensor, $i=1,\ldots, L$, $f_k(\vx_k)$ is the nonlinear function of the state $\vx_k$, $h_k^i(\vx_k)$ is nonlinear measurement function of $\vx_k$ at the $i$th sensor, $\vw_k\in \mathcal {R}^n$ is the uncertain process noise and $\vv_k^i\in \mathcal {R}^{m}$ is the uncertain measurement noise. Assume that $\vw_k$ and $\vv_k^i$ are confined to specified ellipsoidal sets
\begin{eqnarray}
\nonumber\mW_k&=&\{\vw_k: \vw_k^T\mQ_k^{-1}\vw_k\leq1\}\\
\nonumber\mV_k^i&=&\{\vv_k^i: {\vv_k^i}^T({\mR_k^i})^{-1}\vv_k^i\leq1\}
\end{eqnarray}
where $\mQ_k$ and $\mR_k^i$ are the \emph{shape matrix} of the ellipsoids $\mW_k$ and $\mV_k^i$, $i=1,\ldots,L$, respectively. Both of them are known symmetric positive-definite matrices.

Suppose that when the nonlinear functions are linearized, the remainder terms can be bounded by an ellipsoid, respectively. Specifically,
by Taylor's Theorem, $f_k$ and $h_k^i$ can be linearized to
\begin{eqnarray}%
\label{Eqpre_3} f_k(\hat{\vx}_k+\mE_{f_k}\vu_k)=f_k(\hat{\vx}_k)+\mJ_{f_k}\mE_{f_k}\vu_k+\Delta f_k(\vu_k),\\
\label{Eqpre_4} h_k^i(\hat{\vx}_k+\mE_{h_k^i}\vu_k)=h_k^i(\hat{\vx}_k)+\mJ_{h_k^i}\mE_{h_k^i}\vu_k+\Delta h_k^i(\vu_k)
\end{eqnarray}
where $\mJ_{f_k}=\frac{\partial f_k(\vx_k)}{\partial \vx_k}|_{\hat{\vx}_k}$,
$\mJ_{h_k^i}=\frac{\partial h_k^i(\vx_k)}{\partial \vx_k}|_{\hat{\vx}_k}$,
  are Jacobian matrices.
 $\Delta f_k(\vu_k)$ and $\Delta h_k^i(\vu_k)$ are high-order remainders, which can be bounded in an ellipsoid for $\parallel\vu_k\parallel\leq1$, $i=1,\ldots, L$, respectively, i.e.,
\begin{eqnarray}%
\label{Eqpre_5} \Delta f_k(\vu_k)\in  \mathcal {E}_{f_k} &=&\{\vx\in
R^n:(\vx-\ve_{f_k})^T{(\mP_{f_k})}^{-1}(\vx-\ve_{f_k})\leq1\},\\
\label{Eqpre_6}&=&\{\vx\in R^n: \vx=\ve_{f_k}+\mB_{f_k}\Delta_{f_k}, \mP_{f_k}=\mB_{f_k}\mB_{f_k}^T, \parallel\Delta_{f_k}\parallel\leq1\},\\
\label{Eqpre_7} \Delta h_k^i(\vu_k)\in  \mathcal {E}_{h_k^i} &=&\{\vx\in
R^m:(\vx-\ve_{h_k^i})^T{(\mP_{h_k^i})}^{-1}(\vx-\ve_{h_k^i})\leq1\},\\
\label{Eqpre_8}&=&\{\vx\in R^m: \vx=\ve_{h_k^i}+\mB_{h_k^i}\Delta_{h_k^i}, \mP_{h_k^i}=\mB_{h_k^i}\mB_{h_k^i}^T, \parallel\Delta_{h_k^i}\parallel\leq1\},
\end{eqnarray}
where $\ve_{f_k}$ and $\ve_{h_k^i}$ are the centers of the ellipsoids $\mathcal {E}_{f_k}$ and $\mathcal {E}_{h_k^i}$, respectively; $\mP_{f_k}$ and $\mP_{h_k^i}$ are the shape matrices of the ellipsoids $\mathcal {E}_{f_k}$ and $\mathcal {E}_{h_k^i}$, respectively. 
Note that \cite{Wang-Shen-Zhu-Pan16} proposes the Monte Carlo methods for the bounding ellipsoids of the remainders, which can effectively take advantage of the character of the nonlinear functions, and it can obtain the tighter bounding ellipsoids  $\mathcal {E}_{f_k}$ and $\mathcal {E}_{h_k^i}$ to cover the remainders on line. 


The corresponding centralized set-membership information fusion problem can be formulated as follows.
Assume that the initial state $\vx_0$ belongs to a given bounding ellipsoid:
\begin{eqnarray}
\label{Eqpre_9} \mathcal {E}_0^c&=&\{\vx\in
R^n:(\vx-\hat{\vx}_0^c)^T(\mP_0^c)^{-1}(\vx-\hat{\vx}_0^c)\leq1\},
\end{eqnarray}
where $\hat{\vx}_0^c$ is the center of ellipsoid $\mathcal {E}_0^c$, and $\mP_0^c$ is the shape matrix of the ellipsoid $\mathcal {E}_0^c$ which is a known symmetric positive-definite matrix. At time $k$, given that $\vx_k$ belongs to a current bounding ellipsoid:
\begin{eqnarray}
\label{Eqpre_10}  \mathcal {E}_k^c&=&\{\vx\in
R^n:(\vx-\hat{\vx}_k^c)^T(\mP_k^c)^{-1}(\vx-\hat{\vx}_k^c)\leq1\}\\
\label{Eqpre_11}&=&\{\vx\in R^n: \vx=\hat{\vx}_k^c+\mE_k^c\vu_k, \mP_k^c=\mE_k^c{\mE_k^c}^T, \parallel\vu_k\parallel\leq1\},
\end{eqnarray}
where $\hat{\vx}_k^c$ is the center of ellipsoid $\mathcal {E}_k^c$, and $\mP_k^c$ is a known symmetric positive-definite matrix. At next time $k+1$, the fusion center can obtain the measurements $\vy_{k+1}^i$ from the $i$th sensor, $i=1,\ldots,L$.
For the centralized fusion system, the goal of the fusion center is to determine a  prediction ellipsoid $\mathcal {E}_{k+1|k}^c$ and an estimation ellipsoid $\mathcal {E}_{k+1}^c$ at time $k+1$. Firstly, in prediction step, we  look for $\hat{\vx}_{k+1|k}^c$ and $\mP_{k+1|k}^c$ such that the state $\vx_{k+1}$ belongs to
\begin{eqnarray}
\label{Eqpre_112}  \mathcal {E}_{k+1|k}^c&=&\{\vx\in
R^n:(\vx-\hat{\vx}_{k+1|k}^c)^T(\mP_{k+1|k}^c)^{-1}(\vx-\hat{\vx}_{k+1|k}^c)\leq1\}
\end{eqnarray}
whenever I) $\vx_k$ is in $\mathcal {E}_k^c$, II) the process noise  $\vw_k\in\mW_k$, and III) the remainder $\Delta f_k(\vu_k)\in  \mathcal {E}_{f_k}$.
Secondly, in the fusion update step, we look for $\hat{\vx}_{k+1}^c$ and $\mP_{k+1}^c$ such that the state $\vx_{k+1}$ belongs to
\begin{eqnarray}
\label{Eqpre_12}  \mathcal {E}_{k+1}^c&=&\{\vx\in
R^n:(\vx-\hat{\vx}_{k+1}^c)^T(\mP_{k+1}^c)^{-1}(\vx-\hat{\vx}_{k+1}^c)\leq1\}
\end{eqnarray}
whenever I) $\vx_{k+1|k}$ is in $\mathcal {E}_{k+1|k}^c$, II) measurement noises  $\vv_{k+1}^i\in\mV_{k+1}^i$, $i=1,\ldots,L$, and III) the remainders $\Delta h_{k+1}^i(\vu_{k+1})\in  \mathcal {E}_{h_{k+1}^i}$, $i=1,\ldots,L$.

Moreover, we provide a state bounding ellipsoid by minimizing its ``size" at each time which is a function of the shape matrix $\mP$ denoted by $f(\mP)$. If we choose trace function, i.e., $f(\mP)=tr(\mP)$, which means the sum of squares of semiaxes lengths of the ellipsoid $\mathcal {E}$, the other common ``size" of the ellipsoid is $logdet(\mP)$, which corresponds to the volume of the ellipsoid $\mathcal {E}$. In order to emphasize the importance of the interested state vector entry, \cite{Kiselev-Polyak92} proposes an objective of the ellipsoid $\mathcal {E}$ as follows
\begin{eqnarray}
\label{Eqpre_170} f(\mP)=\omega_1\mP_{11}+\omega_2\mP_{22}+ \ldots+\omega_n\mP_{nn}
\end{eqnarray}
where $\omega_i$ is the weight coefficient with $\omega_i>0, \sum_{i=1}^{n}\omega_i=1$, and $\mP_{ii}$ denotes the element in the $i$th row and the $i$th column of the matrix $\mP$, $i=1,\ldots,L$. If the bound of the $i$th entry of the interested state vector is very important, we can give a larger weight to $\omega_i$. When $\omega_i=\frac{1}{n}$, $i=1,\ldots,L$, which means that each entry of the state vector is treated equally, and it is also equivalent to the trace function.

Therefore, we can use  multi-algorithm fusion to obtain  multiple bounding estimated ellipsoids, which squashed along each entry of the state vector as much as possible based on different weighted objective (\ref{Eqpre_170}), then the intersection of these bounding ellipsoids can derive a final state bounding ellipsoid with a smaller size.

\subsection{Problem Formulation for Distributed Fusion}\label{sec_2_2}
In this paper, we also consider $L$-sensor \emph{distributed} estimation fusion for the nonlinear dynamic system  (\ref{Eqpre_1}) and (\ref{Eqpre_2}). The problem is formulated as follows.

At time $k+1$, the $i$th local sensor can use the measurements $\mY_{k+1}^i\triangleq\{\vy_{1}^i, \vy_{2}^i,\ldots,\vy_{k+1}^i\}$ to obtain the bounding ellipsoid $\mathcal {E}_{k+1}^i$ by the single sensor recursive method \cite{Wang-Shen-Zhu-Pan16}. Then, the local estimated ellipsoids $\mathcal {E}_{k+1}^i$ are sent to the fusion center without communication delay for $i=1,\ldots, L$.
Suppose that the initial state $\vx_0$ belongs to a given bounding ellipsoid:
\begin{eqnarray}
\label{Eqpre_13} \mathcal {E}_0^d&=&\{\vx\in
R^n:(\vx-\hat{\vx}_0^d)^T(\mP_0^d)^{-1}(\vx-\hat{\vx}_0^d)\leq1\},
\end{eqnarray}
where $\hat{\vx}_0^d$ is the center of ellipsoid $\mathcal {E}_0^d$, and $\mP_0^d$ is the shape matrix of the ellipsoid $\mathcal {E}_0^d$ which is a known symmetric positive-definite matrix. At time $k$, given that $\vx_k$ belongs to a current bounding ellipsoid:
\begin{eqnarray}
\label{Eqpre_14}  \mathcal {E}_k^d&=&\{\vx\in
R^n:(\vx-\hat{\vx}_k^d)^T(\mP_k^d)^{-1}(\vx-\hat{\vx}_k^d)\leq1\}\\
\label{Eqpre_15}&=&\{\vx\in R^n: \vx=\hat{\vx}_k^d+\mE_k^d\vu_k, \mP_k^d=\mE_k^d{\mE_k^d}^T, \parallel\vu_k\parallel\leq1\},
\end{eqnarray}
where $\hat{\vx}_k^d$ is the center of ellipsoid $\mathcal {E}_k^d$, and $\mP_k^d$ is a known symmetric positive-definite matrix. At next time $k+1$, the fusion center can receive the state bounding ellipsoid of the $i$th sensor
\begin{eqnarray}
\label{Eqpre_16}  \mathcal {E}_{k+1}^i&=&\{\vx\in
R^n:(\vx-\hat{\vx}_{k+1}^i)^T(\mP_{k+1}^i)^{-1}(\vx-\hat{\vx}_{k+1}^i)\leq1\}.
\end{eqnarray}
Firstly, in prediction step, the goal of the fusion center is to determine a state bounding ellipsoid $\mathcal {E}_{k+1|k}^d$, i.e., look for $\hat{\vx}_{k+1|k}^d$ and $\mP_{k+1|k}^d$ such that the state $\vx_{k+1}$ belongs to
\begin{eqnarray}
\label{Eqpre_113}  \mathcal {E}_{k+1|k}^d&=&\{\vx\in
R^n:(\vx-\hat{\vx}_{k+1|k}^d)^T(\mP_{k+1|k}^d)^{-1}(\vx-\hat{\vx}_{k+1|k}^d)\leq1\}
\end{eqnarray}
whenever I) $\vx_k$ is in $\mathcal {E}_k^d$, II) the process noise $\vw_k\in\mW_k$, and III) the remainder $\Delta f_k(\vu_k)\in  \mathcal {E}_{f_k}$.
Secondly, in the fusion update step, we look for $\hat{\vx}_{k+1}^d$ and $\mP_{k+1}^d$ such that the state $\vx_{k+1}$ belongs to
\begin{eqnarray}
\label{Eqpre_17}  \mathcal {E}_{k+1}^d&=&\{\vx\in
R^n:(\vx-\hat{\vx}_{k+1}^d)^T(\mP_{k+1}^d)^{-1}(\vx-\hat{\vx}_{k+1}^d)\leq1\}
\end{eqnarray}
whenever I) $\vx_{k+1}$ is in $\mathcal {E}_{k+1|k}^d$, II) $\vx_{k+1}$ is in $\mathcal {E}_{k+1}^i$, $i=1,\ldots,L$. Moreover, we provide a state bounding ellipsoid by minimizing its ``size" in prediction and update step, respectively.

\section{Centralized Fusion}\label{sec_3}
In this section, we discuss the centralized set-membership estimation fusion, which includes the prediction step and the fusion update step. By taking full advantage of the character of the nonlinear dynamic system and the recent optimization method proposed in \cite{ElGhaoui-Calafiore01} for linear dynamic system, the centralized set-membership estimation fusion  can be achieved by solving an SDP problem, which can be efficiently computed by interior point methods \cite{Vandenberghe-Boyd96} and related softwares \cite{Lofberg04,Sturm99}. Furthermore, the centralized set-membership information filter is derived based on the decoupling technique, which can make further to improve the computation complexity of SDP. The analytical formulae of the state prediction and estimation bounding ellipsoid at time $k+1$ are proposed, respectively.
\subsection{Prediction Step}\label{sec_3-1}
In the prediction step,
the state prediction bounding ellipsoid  at time $k+1$ can be derived as follows.


\begin{lemma}\label{thm_1}
At time $k+1$, based on the state bounding ellipsoid $\mathcal {E}_{k}^c$,
the remainder bounding ellipsoid $\mathcal {E}_{f_k}$ and the noise bounding ellipsoid $\mW_k$, the state prediction bounding ellipsoid $ \mathcal
{E}_{k+1|k}^c=\{\vx:(\vx-\hat{\vx}_{k+1|k}^c)^T(\mP_{k+1|k}^c)^{-1}(\vx-\hat{\vx}_{k+1|k}^c)\leq1\}$
can be obtained by solving the optimization problem in the variables
$\mP_{k+1|k}^c$, $\hat{\vx}_{k+1|k}^c$,  nonnegative scalars $\tau^u\geq0, \tau^w\geq0 , \tau^f\geq0$,
\begin{eqnarray}
\label{Eqpre_18} &&\min~~ f(\mP_{k+1|k}^c) \\[5mm]
\label{Eqpre_19} &&~~\mbox{subject to}~~ -\tau^u\leq0,~ -\tau^w\leq0,~ -\tau^f\leq0,\\[5mm]
\label{Eqpre_20} && \mP_{k+1|k}^c\succ0,\\[5mm]
\label{Eqpre_21}&&\left[\begin{array}{cc}
    \mP_{k+1|k}^c&\Phi_{k+1|k}(\hat{\vx}_{k+1|k}^c)\\[3mm]
    (\Phi_{k+1|k}(\hat{\vx}_{k+1|k}^c))^T& ~~\Xi\\
\end{array}\right]\succeq0,
\end{eqnarray}
where
\begin{eqnarray}
 \label{Eqpre_22} \Phi_{k+1|k}(\hat{\vx}_{k+1|k}^c)&=&[f_k(\hat{\vx}_{k}^c)+\ve_{f_k}-\hat{\vx}_{k+1|k}^c,~\mJ_{f_k}\mE_{k}^c,
 ~\mI,~\mB_{f_k}],\\[3mm]
 \label{Eqpre_23}\Xi &=&\diag(1-\tau^u-\tau^w-\tau^f,\tau^u\mI,\tau^w\mQ_k^{-1},\tau^f\mI),
\end{eqnarray}
 $\mE_{k}^c$ is the Cholesky factorization of $\mP_{k}^c$, i.e, $\mP_{k}^c=\mE_{k}^c(\mE_{k}^c)^T$, $\ve_{f_k}$ and $\mB_{f_k}$ are denoted by (\ref{Eqpre_6}), and $\mJ_{f_k}=\frac{\partial f_k(\vx_k)}{\partial \vx_k}|_{\hat{\vx}_k}$ is Jacobian matrix.
\end{lemma}
\textbf{Proof:} See Appendix.

\begin{remark}
The objective function (\ref{Eqpre_18}) is aimed at minimizing the shape matrix of the predicted ellipsoid, and the constraints (\ref{Eqpre_19})-(\ref{Eqpre_21}) ensure that the true state is contained in the the bounding ellipsoid $ \mathcal
{E}_{k+1|k}$.
\end{remark}

Interestingly, if the objective function is the trace of the shape matrix of the bounding ellipsoid, then the \emph{analytically optimal solution} of the optimization problem (\ref{Eqpre_18})-(\ref{Eqpre_21}) can be achieved for the sate prediction step.

\begin{theorem}\label{cor_1}
If the objective function $f(\mP_{k+1|k}^c)=tr(\mP_{k+1|k}^c)$, then the analytically optimal solution for the state prediction is as follows: 
\begin{eqnarray}
 \label{Eqpre_24}\mP_{k+1|k}^c&=&\frac{\mJ_{f_k}\mP_k^c\mJ_{f_k}^T}{\tau_{opt}^u}+\frac{\mP_{f_k}}{\tau_{opt}^f}+\frac{\mQ_k}{\tau_{opt}^w},\\
 \label{Eqpre_25}\hat{\vx}_{k+1|k}^c&=&f_k(\hat{\vx}_{k}^c)+\ve_{f_k},
 \end{eqnarray}
where
\begin{eqnarray}
 \label{Eqpre_26}\tau_{opt}^u&=&\frac{\sqrt{tr(\mJ_{f_k}\mP_k^c\mJ_{f_k}^T)}}{\sqrt{tr(\mJ_{f_k}\mP_k^c\mJ_{f_k}^T)}+\sqrt{tr(\mQ_k)}+\sqrt{tr(\mP_{f_k})}},\\
 \label{Eqpre_27}\tau_{opt}^f&=&\frac{\sqrt{tr(\mP_{f_k})}}{\sqrt{tr(\mJ_{f_k}\mP_k^c\mJ_{f_k}^T)}+\sqrt{tr(\mQ_k)}+\sqrt{tr(\mP_{f_k})}},\\
 \label{Eqpre_28}\tau_{opt}^w&=&\frac{\sqrt{tr(\mQ_k)}}{\sqrt{tr(\mJ_{f_k}\mP_k^c\mJ_{f_k}^T)}+\sqrt{tr(\mQ_k)}+\sqrt{tr(\mP_{f_k})}},
 \end{eqnarray}
$\mJ_{f_k}=\frac{\partial f_k(\vx_k)}{\partial \vx_k}|_{\hat{\vx}_k}$ is the Jacobian matrix of the nonlinear state function $f_k$ denoted by (\ref{Eqpre_3}),  $\ve_{f_k}$ and $\mP_{f_k}$ are the center and shape matrix of the bounding ellipsoid of the remainder denoted by (\ref{Eqpre_6}), respectively, and $\tau_{opt}^u$, $\tau_{opt}^w$, $\tau_{opt}^f$ are the optimal solution of the decision variables $\tau^u$, $\tau^w$, $\tau^f$, respectively.
\end{theorem}
\textbf{Proof:} See Appendix.
\begin{remark}
When the state equation is linear, there is no the remainder constraint of the nonlinear state equation, i.e., $\mP_{f_k}=0$, it is easy to observe that the optimum ellipsoid derived by  Theorem \ref{cor_1}  coincides with the classical Schweppe bounding ellipsoid \cite{Schweppe68}. 
\end{remark}

\subsection{Fusion update step}\label{sec_3-2}

In the fusion update step,
the state bounding ellipsoid  at time $k+1$ can be derived as follows.
\begin{lemma}\label{thm_2}
At time $k+1$, based on the measurements $\vy_{k+1}^i, i=1,\ldots,L$, the predicted bounding ellipsoid $ \mathcal{E}_{k+1|k}^c$ and the remainder bounding ellipsoids $\mathcal {E}_{h_{k+1}^i}$, $i=1,\ldots,L$, and the noise bounding ellipsoids $\mV_k^i$, $i=1,\ldots,L$,
the centralized state bounding ellipsoid $ \mathcal
{E}_{k+1}^c=\{\vx:(\vx-\hat{\vx}_{k+1}^c)^T(\mP_{k+1}^c)^{-1}(\vx-\hat{\vx}_{k+1}^c)\leq1\}$
can be obtained by solving the optimization problem in the variables
$\mP_{k+1}^c$, $\hat{\vx}_{k+1}^c$,  nonnegative scalars $\tau^u\geq0,
\tau_i^v\geq0,\tau_i^h\geq0$, $i=1,\ldots,L$,
\begin{eqnarray}
\label{Eqpre_49} &&\min~~ f(\mP_{k+1}^c) \\[5mm]
\label{Eqpre_50} &&~~\mbox{subject to}~~ -\tau^u\leq0,~-\tau_i^v\leq0,~ -\tau_i^h\leq0,\\[5mm]
\label{Eqpre_51} && -\mP_{k+1}^c\prec0,\\[5mm]
\label{Eqpre_52}&&\left[\begin{array}{cc}
    -\mP_{k+1}^c&\Phi_{k+1}^c(\hat{\vx}_{k+1}^c)(\Psi_{k+1}^c)_{\bot}\\[3mm]
    (\Phi_{k+1}^c(\hat{\vx}_{k+1}^c)(\Psi_{k+1}^c)_{\bot})^T& ~~-(\Psi_{k+1}^c)_{\bot}^T\Xi(\Psi_{k+1}^c)_{\bot}\\
\end{array}\right]\preceq0,
\end{eqnarray}
where
\begin{eqnarray}
\label{Eqpre_53}  \Phi_{k+1}^c(\hat{\vx}_{k+1}^c)&=&[\hat{\vx}_{k+1|k}^c-\hat{\vx}_{k+1}^c,
 \mE_{k+1|k}^c, \vdots\underbrace{0,\ldots,0}_{L~blocks},\vdots\underbrace{0,\ldots,0}_{L~blocks}], ~~0\in\mathcal {R}^{n,m}, \\[3mm]
\label{Eqpre_54} \Psi_{k+1}^c(\vy_{k+1}^i)&=& [h_{k+1}^i(\hat{\vx}_{k+1|k}^c)+\ve_{h_{k+1}^i}-\vy_{k+1}^i,
\mJ_{h_{k+1|k}^i}\mE_{k+1|k}^c, \\
\nonumber &&~~~~\vdots\underbrace{0,\ldots,\mI,\ldots0}_{the~i-th~ block~ is~ \mI },\vdots
\underbrace{0,\ldots, \mB_{h_{k+1}^i},\ldots,0}_{the~i-th~ block~ is~ \mB_{h_{k+1}^i} }],\\
\label{Eqpre_55} \Psi_{k+1}^c&=&[(\Psi_{k+1}^c(\vy_{k+1}^1))^T,\ldots,(\Psi_{k+1}^c(\vy_{k+1}^L))^T]^T,
\end{eqnarray}
\begin{eqnarray}
\label{Eqpre_56} \Xi =\diag(1-\tau^u-\sum_{i=1}^L\tau_i^v-\sum_{i=1}^L\tau_i^h,\tau^u\mI,
\underbrace{\tau_1^v\mR_{k+1}^{1^{-1}},\ldots,\tau_L^v\mR_{k+1}^{L^{-1}}}_{L~blocks},\underbrace{\tau_1^hI,\ldots,\tau_L^hI}_{L~blocks}).
\end{eqnarray}
$\mE_{k+1|k}^c$ is the Cholesky factorization of $\mP_{k+1|k}^c$, i.e, $\mP_{k+1|k}^c=\mE_{k+1|k}^c(\mE_{k+1|k}^c)^T$, $\hat{\vx}_{k+1|k}^c$ is the center of the predicted bounding ellipsoid $ \mathcal{E}_{k+1|k}^c$, $\ve_{h_{k+1}}$ and $\mB_{h_{k+1}}$ are denoted by $(\ref{Eqpre_8})$ at the time step $k+1$, and $\mJ_{h_{k+1|k}^i}=
 \frac{\partial h_{k+1}^i(\vx_{k+1})}{\partial \vx_{k+1}}|_{\hat{\vx}_{k+1|k}^c}$, $i=1,\ldots,L$,
  are Jacobian matrices.
\end{lemma}
\textbf{Proof:} See Appendix.

Moreover, in order to reduce computation complexity, we can derive an explicit expression of $\mathcal
{E}_{k+1}^c$. 
In Lemma \ref{thm_2}, note that a suitable form of the orthogonal complement of $\Psi_{k+1}^c$ can be chosen as follows
\begin{eqnarray}
\nonumber(\Psi_{k+1}^c)_{\bot}=\left[
                        \begin{array}{cc}
                          -1 & 0 \\
                          \Psi_{21} & \Psi_{22}\\
                        \end{array}
                      \right],
\end{eqnarray}
where
\begin{eqnarray}
\label{Eqpre_122}\Psi_{21}&=&[0,~ (h_{k+1}^1(\hat{\vx}_{k+1|k}^c)-\vy_{k+1}^1)^T,
\ldots,(h_{k+1}^L(\hat{\vx}_{k+1|k}^c)-\vy_{k+1}^L)^T,\\
\nonumber &&\qquad\qquad\qquad\qquad\qquad(\mB_{h_{k+1}^1}^{-1}\ve_{h_{k+1}^1})^T,\ldots,(\mB_{h_{k+1}^L}^{-1}\ve_{h_{k+1}^L})^T]^T,\\
\label{Eqpre_118}\Psi_{22}&=&\left[
            \begin{array}{ccccc}
              (\mE_{k+1|k}^c)^{-1} & 0 & 0 & \cdots & 0 \\
              -\mJ_{h_{k+1|k}^1} & \mI & 0 & \vdots & 0 \\
              -\mJ_{h_{k+1|k}^2} & 0 & \mI & \vdots & 0 \\
              \vdots & \vdots & \vdots & \vdots & \vdots \\
              -\mJ_{h_{k+1|k}^L} & 0 & 0 & \vdots & \mI \\
              0 & -\mB_{h_{k+1}^1}^{-1} & 0 & \vdots & 0 \\
              0 & 0 & -\mB_{h_{k+1}^2}^{-1} & \vdots & 0 \\
              \vdots & \vdots & \vdots & \vdots & \vdots \\
              0 & 0 & 0 & \cdots & -\mB_{h_{k+1}^L}^{-1} \\
            \end{array}
          \right].
\end{eqnarray}
If we denote
\begin{eqnarray}
\nonumber\Xi&=&\diag(\Xi_{11},\Xi_{22}),\\
\label{Eqpre_119}\Xi_{11}&=&1-\tau^u-\sum_{i=1}^L\tau_i^v-\sum_{i=1}^L\tau_i^h,\\
\label{Eqpre_120}\Xi_{22}&=&\diag(\tau^u\mI,\underbrace{\tau_1^v\mR_{k+1}^{1^{-1}},
\ldots,\tau_L^v\mR_{k+1}^{L^{-1}}}_{L~blocks},\underbrace{\tau_1^hI,\ldots,\tau_L^hI}_{L~blocks}),
\end{eqnarray}
then Equation (\ref{Eqpre_52}) is equivalent to the following form by reordering of the blocks
\begin{eqnarray}
\label{Eqpre_80}&&\left[
  \begin{array}{ccc}
    \mP_{k+1}^c &\hat{\vx}_{k+1}^c-\hat{\vx}_{k+1|k}^c & \mB \\
    (\hat{\vx}_{k+1}^c-\hat{\vx}_{k+1|k}^c)^T & \Xi_{11}+\Psi_{21}^T\Xi_{22}\Psi_{21} & \Psi_{21}^T\Xi_{22}\Psi_{22} \\
    \mB^T & \Psi_{22}^T\Xi_{22}\Psi_{21} & \Psi_{22}^T\Xi_{22}\Psi_{22} \\
  \end{array}
\right]\succeq0,\\
&&\label{Eqpre_121}\mB=[\mI ~\underbrace{0,\ldots, 0}_{L~ blocks}],
\end{eqnarray}
where $\mI$ and 0 have compatible dimensions. Moreover, the decoupled fusion update step is given in the following theorem.

\begin{theorem}\label{cor_2}
Consider the optimization problem in the variables $\tau^u,\tau_i^v,\tau_i^h,i=1,\ldots,L$
\begin{eqnarray}
\label{Eqpre_57} &&\min~~ f(\mB(\Psi_{22}^T\Xi_{22}\Psi_{22})^{-1}\mB^T) \\[5mm]
 \label{Eqpre_58}&&~~\mbox{subject to}~~ -\tau^u\leq0,~-\tau_i^v\leq0,~ -\tau_i^h\leq0,\\[5mm]
 \label{Eqpre_59}&&~~~\left[
        \begin{array}{cc}
          \Xi_{11}+\Psi_{21}^T\Xi_{22}\Psi_{21} & \Psi_{21}^T\Xi_{22}\Psi_{22} \\
          \Psi_{22}^T\Xi_{22}\Psi_{21} & \Psi_{22}^T\Xi_{22}\Psi_{22} \\
        \end{array}
      \right]\succeq0,
\end{eqnarray}
where $\Psi_{21}$, $\Psi_{22}$, $ \Xi_{11}$, $ \Xi_{22}$, $\mB$ are denoted by (\ref{Eqpre_122}), (\ref{Eqpre_118}), (\ref{Eqpre_119}), (\ref{Eqpre_120}), (\ref{Eqpre_121}), respectively.
If the above problem is feasible, then there exists an optimal ellipsoid. The shape matrix and center of the optimal fusion update ellipsoid $\mathcal{E}_{k+1}^c$  are given by
\begin{eqnarray}
\label{Eqpre_62}\mP_{k+1}^{c^{-1}}&=&\tau_{opt}^u\mP_{k+1|k}^{c^{-1}}+\sum_{i=1}^L\mJ_{h_{k+1|k}^i}^T\left(\frac{\mR_{k+1}^i}{\tau_{opt_i}^v}+\frac{\mP_{h_{k+1}^i}}{\tau_{opt_i}^h}\right)^{-1}
\mJ_{h_{k+1|k}^i}\\
\label{Eqpre_180}\hat{\vx}_{k+1}^c&=&\hat{\vx}_{k+1|k}^c+\sum_{i=1}^L\tau_{opt_i}^v\mK_{k+1}^i(\vy_{k+1}^i-h_{k+1}^i(\hat{\vx}_{k+1|k}^c))-\mC_{k+1},
\end{eqnarray}
where
\begin{eqnarray}
\mK_{k+1}^i&=&\mP_{k+1}^c\mJ_{h_{k+1|k}^i}^T\mR_{k+1}^{i^{-1}}-\mM_1\mM_2\mJ_{h_{k+1|k}^i}^T\mR_{k+1}^{i^{-1}}(\tau_{opt_i}^v\mR_{k+1}^{i^{-1}}+\tau_{opt_i}^h\mP_{h_{k+1}^{i}}^{-1})^{-1}\tau_{opt_i}^v\mR_{k+1}^{i^{-1}}\\
\nonumber \mC_{k+1}&=& \mM_1\mM_2\left(\sum_{i=1}^L\tau_{opt_i}^v\mJ_{h_{k+1|k}^i}^T\mR_{k+1}^{i^{-1}}(\tau_{opt_i}^v\mR_{k+1}^{i^{-1}}+\tau_{opt_i}^h\mP_{h_{k+1}^{i}}^{-1})^{-1}
\tau_{opt_i}^h\mP_{h_{k+1}^i}^{-1}\ve_{h_{k+1}^i}\right)
\end{eqnarray}
\begin{eqnarray}
\nonumber \mM_1&=&\left(\tau_{opt}^u\mP_{k+1|k}^{c^{-1}}+\sum_{i=1}^L\tau_{opt_i}^v\mJ_{h_{k+1|k}^i}^T\mR_{k+1}^{i^{-1}}\mJ_{h_{k+1|k}^i}\right)^{-1}\\
\nonumber \mM_2&=&\mI+\sum_{i=1}^L\tau_{opt_i}^v\mJ_{h_{k+1|k}^i}^T\mR_{k+1}^{i^{-1}}(\tau_{opt_i}^v\mR_{k+1}^{i^{-1}}+\tau_{opt_i}^h\mP_{h_{k+1}^{i}}^{-1})^{-1}\tau_{opt_i}^v\mR_{k+1}^{i^{-1}}\mJ_{h_{k+1|k}^i}\mP_{k+1}^c,
\end{eqnarray}
$\mJ_{h_{k+1|k}^i}=\frac{\partial h_{k+1}^i(\vx_{k+1})}{\partial \vx_{k+1}}|_{\hat{\vx}_{k+1|k}^c}$ $i=1,\ldots,L$ are the Jacobian matrices of the nonlinear measurement function $h_{k+1}^i$ denoted by (\ref{Eqpre_4}), $\ve_{h_{k+1}^i}$ and $\mP_{h_{k+1}^i}$  are the center and shape matrix of the bounding ellipsoid of the remainder denoted by (\ref{Eqpre_8}), respectively, and $\tau_{opt_i}^v,\tau_{opt_i}^h$ are the optimal solutions of the decision variables $\tau_i^v,\tau_i^h,i=1,\ldots,L$, in the optimization problem (\ref{Eqpre_57})--(\ref{Eqpre_59}).
\end{theorem}

\textbf{Proof:} See Appendix.

\begin{remark}
Here, we call the equations (\ref{Eqpre_62})--(\ref{Eqpre_180}) \emph{centralized set-membership information filter}, which has following characters:
\begin{itemize}
  \item Similar to the information filter \cite{BarShalom-Li-Kirubarajan01}, $\mJ_{h_{k+1|k}^i}^T\left(\frac{\mR_{k+1}^i}{\tau_{opt_i}^v}+\frac{\mP_{h_{k+1}^i}}{\tau_{opt_i}^h}\right)^{-1}
  \mJ_{h_{k+1|k}^i}$ and $\mK_{k+1}^i$  in (\ref{Eqpre_62})--(\ref{Eqpre_180}) can be taken as the update information matrix and the gain matrix provided by the $i$-th sensor for the estimator, respectively. $\tau_{opt_i}^v, i=1,\ldots,L$ are the fusion weights.

  \item $\mC_{k+1}$ is the nonlinear correction term of the state update estimation, which relies on the nonlinear measurement functions $h_{k+1}^i$, $i=1,\ldots,L$.

   \item When the measurement equations are linear,  there are no the remainder constraints, i.e., $\mP_{h_{k+1}^i}=0$, it is easy to observe that the optimum ellipsoid derived by the Theorem \ref{cor_2}  also similar to the classical Schweppe bounding ellipsoid \cite{Schweppe68}.
\end{itemize}

\end{remark}
\begin{remark}
If $f(\mP)=tr(\mP)$ and $\Psi_{k+1}^c$ is full-rank, then the optimization problem (\ref{Eqpre_49})-(\ref{Eqpre_52}) in Lemma \ref{thm_2} is an SDP problem, the dimension of the constraint matrix (\ref{Eqpre_52}) is $\mM=n+(n+2mL+1-mL)=mL+2n+1$ and the number of decision variables is $\mN=\frac{n(n+1)}{2}+n+2L+1$, where $n, m$ and $L$ are the dimensions of the state, the measurement and the number of sensors, respectively. Moreover, if we use a general-purpose primal-dual interior-point algorithm to solve it, then the computation complexity of the problem is $\emph{O}(\mM^2\mN^2)$, see \cite{Vandenberghe-Boyd96}. Therefore, in our case, the computation complexity is $\emph{O}(n^6)$ if $n> mL$, otherwise, it is $\emph{O}(m^2L^4)$.

As described in \cite{Nesterov-Nemirovski94}, we can use a path-following interior-point method to solve (\ref{Eqpre_57})-(\ref{Eqpre_59}) in Theorem \ref{cor_2}. A tedious but straightforward computation shows the practical complexity can be assumed to be $\emph{O}(n^3L+m^3L^4)$, which implies an $\emph{O}(n^3)$ dependence on the size of the state $\vx$, and $\emph{O}(m^3L^4)$ dependence on the number of the sensor. Therefore, for the $\emph{fixed}$ number of sensors, the complexity of the decoupled problem (\ref{Eqpre_57}) improves upon that of the coupled one (\ref{Eqpre_49}) by a factor of $\emph{O}(n^3)$.
\end{remark}

The centralized set membership information fusion algorithm can be summarized as follows.
\begin{algorithm}[Centralized set membership information fusion algorithm]\label{alg_1}
~\\
\begin{itemize}
\item Step~1: (Initialization step) Set $k=0$ and  initial values $(\hat{\vx}_0,\mP_0)$ such that $\vx_0\in\mathcal {E}_0$.
\item Step~2: (Bounding step) Take samples $\vu_k^1,\ldots,\vu_k^N$ from the sphere $||\vu_k||\leq 1$, and then determine two bounding ellipsoids to cover the remainders  $\Delta f_k$ by (\ref{Eqpre_5})-(\ref{Eqpre_6}).
\item Step~3: (Prediction step \cite{Wang-Shen-Zhu-Pan16}) Optimize the center and shape matrix of the state prediction ellipsoid $(\hat{\vx}_{k+1|k}^c,\mP_{k+1|k}^c)$ such that $\vx_{k+1|k}^c\in\mathcal {E}_{k+1|k}^c$ by (\ref{Eqpre_18})-(\ref{Eqpre_21}) or (\ref{Eqpre_24})-(\ref{Eqpre_25}).
\item Step~4: (Bounding step \cite{Wang-Shen-Zhu-Pan16}) Take samples $\vu_{k+1|k}^1,\ldots,\vu_{k+1|k}^N$ from the sphere $||\vu_{k+1|k}||\leq 1$, and then determine one bounding ellipsoid to cover the remainder  $\Delta h_{k+1|k}^i$, $i=1,\ldots,L$, by  (\ref{Eqpre_7})-(\ref{Eqpre_8}).
 \item Step~5: (Fusion update step) Optimize the center and shape matrix of the state estimation ellipsoid $(\hat{\vx}_{k+1}^c,\mP_{k+1}^c)$ such that $\vx_{k+1}^c\in\mathcal {E}_{k+1}^c$ by solving the optimization problem (\ref{Eqpre_49})-(\ref{Eqpre_52}) or (\ref{Eqpre_57})-(\ref{Eqpre_59}).
 \item Step~6: Set $k=k+1$ and go to step~2.
\end{itemize}
\end{algorithm}

\section{Distributed Fusion}\label{sec_4}
In this section, in order to reduce the computation burden of the fusion center and improve
the reliability, robustness, and survivability of the fusion system \cite{Li-Zhu-Wang-Han03}, the distributed set-membership estimation fusion method is derived by fusing the state bounding ellipsoids, which are sent from the local sensors and using the character of the nonlinear state function. Since the state prediction step of the distributed fusion is completely same as that of the centralized fusion, we only discuss the fusion update step of the distributed fusion. In addition, the distributed set-membership information fusion formula can also be achieved by the decoupling technique. The main results are summarized to Lemma \ref{thm_3} and Theorem \ref{cor_3}. The proofs are also given in Appendix.



\begin{lemma}\label{thm_3}
At time $k+1$, based on the prediction bounding ellipsoids $\mathcal {E}_{k+1|k}^d$ and the estimation bounding ellipsoids of single sensors $\mathcal {E}_{k+1}^i$, $i=1,\ldots,L$, the distributed state bounding ellipsoid $ \mathcal
{E}_{k+1}^d=\{\vx:(\vx-\hat{\vx}_{k+1}^d)^T(\mP_{k+1}^d)^{-1}(\vx-\hat{\vx}_{k+1}^d)\leq1\}$
can be obtained by solving the optimization problem in the variables
$\mP_{k+1}^d$, $\hat{\vx}_{k+1}^d$,  nonnegative scalars $\tau^u\geq0,
\tau_i^y\geq0$, $i=1,\ldots,L$,
\begin{eqnarray}
\label{Eqpre_81} &&\min~~ f(\mP_{k+1}^d) \\[5mm]
\label{Eqpre_82}&&~~\mbox{subject to}~~ -\tau^u\leq0,~-\tau_i^y\leq0\\[5mm]
\label{Eqpre_83} && -\mP_{k+1}^d\prec0,\\[5mm]
\label{Eqpre_84}&&\left[\begin{array}{cc}
    -\mP_{k+1}^d&\Phi_{k+1}^d\\[3mm]
    (\Phi_{k+1}^d)^T& ~~-\Xi-\Pi\\
\end{array}\right]\preceq0,
\end{eqnarray}
where
\begin{eqnarray}
 \label{Eqpre_85} \Phi_{k+1}^d&=&[\hat{\vx}_{k+1|k}^d-\hat{\vx}_{k+1}^d,\mE_{k+1|k}^d],\\[3mm]
 \label{Eqpre_86} \Phi_{k+1}^i&=&[\hat{\vx}_{k+1|k}^d-\hat{\vx}_{k+1}^i,\mE_{k+1|k}^d], \\[3mm]
 \label{Eqpre_87}\Pi&=& \sum_{i=1}^L\tau_i^y(\Phi_{k+1}^i)^T(\mP_{k+1}^i)^{-1}\Phi_{k+1}^i,\\
\label{Eqpre_88}\Xi &=&\diag(1-\tau^u-\sum_{i=1}^L\tau_i^y,\tau^uI),
\end{eqnarray}
 $\mE_{k+1|k}^d$ is the Cholesky factorization of $\mP_{k+1}^d$, i.e, $\mP_{k+1}^d=\mE_{k+1|k}^d(\mE_{k+1|k}^d)^T$.
\end{lemma}
\textbf{Proof:} See Appendix.

\begin{remark}
Compared with the centralized fusion in Lemma \ref{thm_2}, it can be seen that the dimension of the constraint matrix (\ref{Eqpre_84}) is $\mM=2n+1$ independent of the number of the sensors and the number of decision variables is $\mN=\frac{n(n+1)}{2}+n+L+1$. However, the dimension of the constraint matrix (\ref{Eqpre_52}) is $\mM=mL+2n+1$, and the number of the decision variables is $\mN=\frac{n(n+1)}{2}+n+2L+1$. Therefore, the distributed fusion can decrease much more computation burden of the fusion center.
\end{remark}

Note that  (\ref{Eqpre_84}) can be rewritten to
\begin{eqnarray}
\label{Eqpre_106}\left[\begin{array}{ccc}
    \mP_{k+1}^d&\hat{\vx}_{k+1|k}^d-\hat{\vx}_{k+1}^d&\mE_{k+1|k}^d\\[3mm]
\label{Eqpre_0106}  (\hat{\vx}_{k+1|k}^d-\hat{\vx}_{k+1}^d)^T&\Upsilon_{11} & \Upsilon_{12}\\[3mm]
    (\mE_{k+1|k}^d)^T&\Upsilon_{12}^T & \Upsilon_{22}\\
\end{array}\right]\succeq0,
\end{eqnarray}
where
\begin{eqnarray}
\label{Eqpre_123}\Upsilon_{11}&=&1-\tau^u-\sum_{i=1}^L\tau_i^y+\sum_{i=1}^L\tau_i^y(\hat{\vx}_{k+1|k}^d-\hat{\vx}_{k+1}^i)^T\mP_{k+1}^{i^{-1}}(\hat{\vx}_{k+1|k}^d-\hat{\vx}_{k+1}^i),\\
\label{Eqpre_124}\Upsilon_{12}&=&\sum_{i=1}^L\tau_i^y(\hat{\vx}_{k+1|k}^d-\hat{\vx}_{k+1}^i)^T\mP_{k+1}^{i^{-1}}\mE_{k+1|k}^d,\\
\label{Eqpre_125}\Upsilon_{22}&=&\tau^uI+\sum_{i=1}^L\tau_i^y\mE_{k+1|k}^{d^T}\mP_{k+1}^{i^{-1}}\mE_{k+1|k}^d.
\end{eqnarray}

Moreover, we can derive an analytical formula for the shape matrix and the center of the bounding ellipsoid $\mathcal {E}_{k+1|k}^d$  as follows.

\begin{theorem}\label{cor_3}
Consider the convex optimization problem in the variables $\tau^u, \tau_i^y$, $i=1,\ldots,L$,
\begin{eqnarray}
 \label{Eqpre_107}&&\min~~ f(\mE_{k+1|k}^d\Upsilon_{22}^{-1}\mE_{k+1|k}^{d^T}) \\[5mm]
 \label{Eqpre_108}&&~~\mbox{subject to}~~ -\tau^u\leq0,~-\tau_i^y\leq0\\[5mm]
 \label{Eqpre_109}&&~~~\left[
        \begin{array}{cc}
          \Upsilon_{11} & \Upsilon_{12} \\
          \Upsilon_{12}^T & \Upsilon_{22} \\
        \end{array}
      \right]\succeq0,
\end{eqnarray}
where $\mE_{k+1|k}^d$ is the Cholesky factorization of $\mP_{k+1}^d$, i.e, $\mP_{k+1}^d=\mE_{k+1|k}^d(\mE_{k+1|k}^d)^T$, and $ \Upsilon_{11}$,$ \Upsilon_{12}$, $ \Upsilon_{22}$ are denoted by (\ref{Eqpre_123})-(\ref{Eqpre_125}), respectively.
If the above problem is feasible, then there exists an optimal bounding ellipsoid, and the shape matrix and center of the optimal bounding ellipsoid $\mathcal {E}_{k+1|k}^d$ are given by
\begin{eqnarray}
\label{Eqpre_110}(\mP_{k+1}^d)^{-1} &=&\tau_{opt}^u(\mP_{k+1|k}^{d})^{-1}+\sum_{i=1}^L\tau_{opt_i}^y(\mP_{k+1}^{i})^{-1}\\
\label{Eqpre_111}  \hat{\vx}_{k+1}^d&=& \hat{\vx}_{k+1|k}^d+\sum_{i=1}^L\tau_{opt_i}^y\mP_{k+1}^d(\mP_{k+1}^{i})^{-1}(\hat{\vx}_{k+1}^i-\hat{\vx}_{k+1|k}^d),
\end{eqnarray}
where $\tau_{opt_i}^y$ is the optimal solution of the decision  variable $\tau_i^y, i=1,\ldots,L$, respectively.
\end{theorem}
\begin{proof}
The proof is similar to Theorem \ref{cor_2}.
\end{proof}
\begin{remark}
We call the equations (\ref{Eqpre_110})--(\ref{Eqpre_111})  \emph{distributed set-membership information filter}. In (\ref{Eqpre_110})--(\ref{Eqpre_111}),
$\tau_{opt_i}^y(\mP_{k+1}^{i})^{-1}$ and $\mP_{k+1}^d(\mP_{k+1}^{i})^{-1}$ can be taken as the update information matrix and the gain matrix provided by the $i$-th sensor for the estimator, respectively, and $\tau_{opt_i}^y, i=1,\ldots,L$ are the fusion weights.
%
\end{remark}

The distributed set membership information fusion algorithm can be summarized as follows.
\begin{algorithm}[Distributed set membership information fusion algorithm]\label{alg_2}
~\\
\begin{itemize}
\item Step~1: (Initialization step) Set $k=0$ and  initial values $(\hat{\vx}_0,\mP_0)$ such that $\vx_0\in\mathcal {E}_0$.
\item Step~2: (Bounding step \cite{Wang-Shen-Zhu-Pan16}) Take samples $\vu_k^1,\ldots,\vu_k^N$ from the sphere $||\vu_k||\leq 1$, and then determine a bounding ellipsoid to cover the remainders  $\Delta f_k$ by (\ref{Eqpre_5})-(\ref{Eqpre_6}).
\item Step~3: (Prediction step) Optimize the center and shape matrix of the state prediction ellipsoid $(\hat{\vx}_{k+1|k}^d,$ $\mP_{k+1|k}^d)$ such that $\hat{\vx}_{k+1|k}^d\in\mathcal {E}_{k+1|k}^d$ by solving the optimization problem (\ref{Eqpre_18})-(\ref{Eqpre_21}) or (\ref{Eqpre_24})-(\ref{Eqpre_25}).
\item Step~4: (Fusion update step) Optimize the center and shape matrix of the state estimation ellipsoid $(\hat{\vx}_{k+1}^d,\mP_{k+1}^d)$ such that $\vx_{k+1}^d\in\mathcal {E}_{k+1}^d$ by solving the optimization problem  (\ref{Eqpre_81})-(\ref{Eqpre_84}) or (\ref{Eqpre_107})-(\ref{Eqpre_109}) based on the state prediction bounding ellipsoids $\mathcal {E}_{k+1|k}^d$ and bounding ellipsoids of single sensors $\mathcal {E}_{k+1}^i$, $i=1,\ldots,L$.
 \item Step~5: Set $k=k+1$ and go to step~2.
\end{itemize}
\end{algorithm}

\begin{remark}
In target tracking, whether it is the distributed fusion or the centralized fusion, if the measurement only contain range and angle, the boundary sampling method \cite{Wang-Shen-Zhu-Pan16} can be used to drive the bounding ellipsoid of the remainders with less computation complexity. Therefore, the bounding steps of Algorithm \ref{alg_1} and Algorithm \ref{alg_2} can be computed efficiently. Finally, the set-membership information fusion formulae are summarized in Table \ref{tab_1}.
\end{remark}

\begin{table}
 \caption{Set-Membership Information Fusion Formulae}
{\tiny \begin{tabular}{|c|c|c|}
  \hline
  Fusion method & Centralized set membership information fusion algorithm& Distributed set membership information fusion algorithm\\\hline
  Nonlinear  &  $\vx_{k+1}=f_k(\vx_k)+\vw_k$\qquad\qquad\qquad\qquad & $\vx_{k+1}=f_k(\vx_k)+\vw_k$\\
  model&$\vy_k^i=h_k^i(\vx_k)+\vv_k^i,~i=1,\ldots, L $& $\hat{\vx}_{k+1}^i=\vx_{k+1}-\hat{\vx}_{k+1|k}^d-\mE_{k+1|k}^d\vu_{k+1|k}+\hat{\vx}_{k+1}^i$\\\hline
  Noise bounds &  $\mW_k=\{\vw_k: \vw_k^T\mQ_k^{-1}\vw_k\leq1\}$ & $\mW_k=\{\vw_k: \vw_k^T\mQ_k^{-1}\vw_k\leq1\}$\\
  ~&$\mV_k^i=\{\vv_k^i: {\vv_k^i}^T({\mR_k^i})^{-1}\vv_k^i\leq1\}$&$\mathcal {E}=\{\vu_{k+1|k}: \parallel\vu_{k+1|k}\parallel\leq1\}$\\\hline
  Remainder bounds &  Methods in \cite{Scholte-Campbell03} or \cite{Wang-Shen-Zhu-Pan16}& Methods in \cite{Scholte-Campbell03} or \cite{Wang-Shen-Zhu-Pan16} \\\hline
  Data received &  $\vy_k^1,\ldots,\vy_k^L$ & $\hat{\vx}_{k}^1,\ldots,\hat{\vx}_{k}^L$\\\hline
  Optimum weights &  SDP (\ref{Eqpre_57})-(\ref{Eqpre_59})& SDP (\ref{Eqpre_107})-(\ref{Eqpre_109})\\\hline
~  & $\mP_{k+1}^{c^{-1}}=\tau_{opt}^u\mP_{k+1|k}^{c^{-1}}\qquad\qquad\qquad\qquad\qquad\qquad\qquad\qquad  $(\ref{Eqpre_62})
 &  ~\\
  Information&$+\sum_{i=1}^L\mJ_{h_{k+1|k}^i}^T\left(\frac{\mR_{k+1}^i}{\tau_{opt_i}^v}+\frac{\mP_{h_{k+1}^i}}{\tau_{opt_i}^h}\right)^{-1}\mJ_{h_{k+1|k}^i}$&
$(\mP_{k+1}^d)^{-1}=\tau_{opt}^u(\mP_{k+1|k}^{d})^{-1}+\sum_{i=1}^L\tau_{opt_i}^y(\mP_{k+1}^{i})^{-1}$ \qquad (\ref{Eqpre_110})\\
  filter fuser& $\hat{\vx}_{k+1}^c=\hat{\vx}_{k+1|k}^c\qquad\qquad\qquad\qquad \qquad\qquad\qquad\qquad \qquad\qquad$ (\ref{Eqpre_180})&   $\hat{\vx}_{k+1}^d=\hat{\vx}_{k+1|k}^d\qquad\qquad\qquad\qquad \qquad\qquad\qquad\qquad \qquad\qquad\qquad\qquad$ (\ref{Eqpre_111}) \\
  ~&$+\sum_{i=1}^L\tau_{opt_i}^v\mK_{k+1}^i(\vy_{k+1}^i-h_{k+1}^i(\hat{\vx}_{k+1|k}^c))-\mC_{k+1}$&
  $+\sum_{i=1}^L\tau_{opt_i}^y\mP_{k+1}^d(\mP_{k+1}^{i})^{-1}(\hat{\vx}_{k+1}^i-\hat{\vx}_{k+1|k}^d)$\\
  \hline
  \end{tabular}}\label{tab_1}
\end{table}
\begin{remark}
As far as multi-algorithm fusion for nonlinear dynamic systems is concerned, the multiple bounding ellipsoids can be constructed to minimize the size of the state bounding ellipsoid by complementary advantages of multiple parallel algorithms. Specifically, one can use multiple parallel Algorithm \ref{alg_1} or \ref{alg_2}  with differently weighted objectives in (\ref{Eqpre_170}), where the larger $\omega_j$ emphasizes the $j$th entry of the estimated state vector, then the intersection of these bounding ellipsoids can achieve a tighter bounding ellipsoid that containing the true state in fusion center.
\end{remark}
\section{Numerical examples }\label{sec_5}
In this section, we provide an example to compare the performance of the centralized fusion with that of the distributed fusion. Moreover, we also use the multi-algorithm fusion to further reduce the estimation error bound based on the different weighted objective (\ref{Eqpre_170}).

Consider a common tracking system with bounding noise and there are two sensors track a same target in different position. The state contain position and velocity of $x$ and $y$ directions. Here, the dynamic system equations is as follows \cite{BarShalom-Li-Kirubarajan01}:
\begin{eqnarray}%
\label{Eqpre_114}\vx_{k+1}&=&\left[
                         \begin{array}{cccc}
                           1 & 0& T& 0\\
                           0 & 1 & 0 & T \\
                           0 & 0 & 1 & 1 \\
                           0 & 0 & 0 & 1 \\
                         \end{array}
                       \right]\vx_k+\vw_k,\\[3mm]
\label{Eqpre_115}\vy_k^i&=&\left[
                      \begin{array}{c}
                        \sqrt{(\vx_k(1)-\vz_k^i(1))^2+(\vx_k(2)-\vz_k^i(2))^2} \\[3mm]
                        arctan\left(\frac{\vx_k(2)-\vz_k^i(2)}{\vx_k(1)-\vz_k^i(1)}\right) \\
                      \end{array}
                    \right]+\vv_k^i,\\
                    \nonumber && ~for~ i=1,2.
\end{eqnarray}
where  $T$ is the time sampling interval with $T=1$. $\vz_k^i=[\vz_k^i(1)~\vz_k^i(2)]^T$ is the position of the $i$th sensor,
where $\vz_k^1=[525~525]^T$ and $\vz_k^2=[524~524]^T$. Moreover, the process noise $\vw_k$ and measurement noise $\vv_k$ are taking value in specified ellipsoidal sets
\begin{eqnarray}
\nonumber\mW_k&=&\{\vw_k: \vw_k^T\mQ_k^{-1}\vw_k\leq1\}\\
\nonumber\mV_k^i&=&\{\vv_k^i: \vv_k^T{\mR_k^i}^{-1}\vv_k\leq1\}.
\end{eqnarray}
where
\begin{eqnarray}
\nonumber \mQ_k&=&\sigma^2\left[
                          \begin{array}{cccc}
                            \frac{T^3}{3} & 0 & \frac{T^2}{2} & 0 \\
                            0 & \frac{T^3}{3} & 0 & \frac{T^2}{2} \\
                            \frac{T^2}{2} & 0 & T & 0 \\
                            0 & \frac{T^2}{2} & 0 & T \\
                          \end{array}
                        \right]\\
 \nonumber \mR_k^i&=&\left[
                   \begin{array}{cc}
                     0.01 & 0\\
                     0 & 25\\
                   \end{array}
                 \right].
\end{eqnarray}
The target acceleration is $\sigma^2=1$.
In the example, the target starts at the point $(120,120)$ with a velocity of $(6,6)$.

The center and the shape matrix of the initial bounding ellipsoid are $\hat{\vx}_0=\left[
 \begin{array}{cccc}
    120& 120 & 6 &6\\
 \end{array}
\right]^T$,
\begin{eqnarray}
\nonumber \mP_0=\left[
                  \begin{array}{cccc}
                    100 & 0 & 0 & 0 \\
                    0 & 100 & 0 & 0 \\
                    0 & 0 & 30 & 0 \\
                    0 & 0 & 0 & 30 \\
                  \end{array}
                \right],
\end{eqnarray}
respectively.

In order to simulate the performance of the center fusion and distributed fusion, we assume the process noise measurement noise are truncated Gaussian with zeros mean and covariance $\frac{\mQ_k}{9}$ and $\frac{\mR_k^i}{9}$ on the ellipsoidal sets, respectively.
From the description of the above, we can use sensor 1 (SMF1), sensor 2 (SMF2), the centralized fusion (CSMF) and distributed fusion (DSMF) to calculate the error bound with $w=[\frac{1}{4} ~\frac{1}{4} ~\frac{1}{4}~ \frac{1}{4}]$ in (\ref{Eqpre_170}), respectively, moreover, we also use the multi-algorithm fusion (MSMF) to produce the error bound based on the different weight coefficient with  $w_1=[\frac{19}{25} ~\frac{2}{25} ~\frac{2}{25}~ \frac{2}{25}]$,  $w_2=[\frac{2}{25} ~\frac{19}{25} ~\frac{2}{25}~ \frac{2}{25}]$, $w_3=[\frac{2}{25} ~\frac{2}{25} ~\frac{19}{25}~ \frac{2}{25}]$,  $w_4=[\frac{2}{25} ~\frac{2}{25} ~\frac{2}{25}~ \frac{19}{25}]$, where the \textbf{error bound} of the $i$th entry of the state $\vx_{k+1}$ can be calculated by projecting the ellipsoid along the $i$th output direction.

The following simulation results are under Matlab R2012a with YALMIP.

Figs. \ref{fig_01}-\ref{fig_04} present a comparison of the error bounds along position and velocity direction for sensors 1, 2 using Algorithm \ref{alg_1} (L=1) and for the fusion center using the centralized fusion Algorithm \ref{alg_1} (L=2) and the distributed fusion Algorithm \ref{alg_2} (L=2) and the multi-algorithm fusion, respectively.

From Figs. \ref{fig_01}-\ref{fig_04}, we can observe the following phenomenon:
\begin{itemize}
  \item The performance of the centralized fusion and the distributed fusion is better than that of sensors.
  \item The performance of the centralized fusion is better than that of the distributed fusion along $x$ and $y$ position direction in Figs. \ref{fig_01}-\ref{fig_02}, but the distributed fusion performs slightly better than centralized fusion along $x$ and $y$ velocity direction in Figs. \ref{fig_03}-\ref{fig_04}. The reasons may be that the optimal bounding ellipsoid cannot be obtained for the nonlinear dynamic system, and the error bound of the state vector is calculated by minimizing trace of the shape matrix of the bounding state ellipsoid rather than minimizing the error bounds along position and velocity directions, respectively.
  \item The performance of the multi-algorithm fusion is significantly better than that of the other methods along position and velocity direction. Since it extract the useful information of each entry of the state vector by the differently weighted objectives. Then the intersection fusion of these estimation ellipsoids can sufficiently take advantage of the information of each sensor, which yields a tighter state bounding ellipsoidal.
\end{itemize}

\begin{figure}[h]
\vbox to 6cm{\vfill \hbox to \hsize{\hfill
\scalebox{0.6}[0.59]{\includegraphics{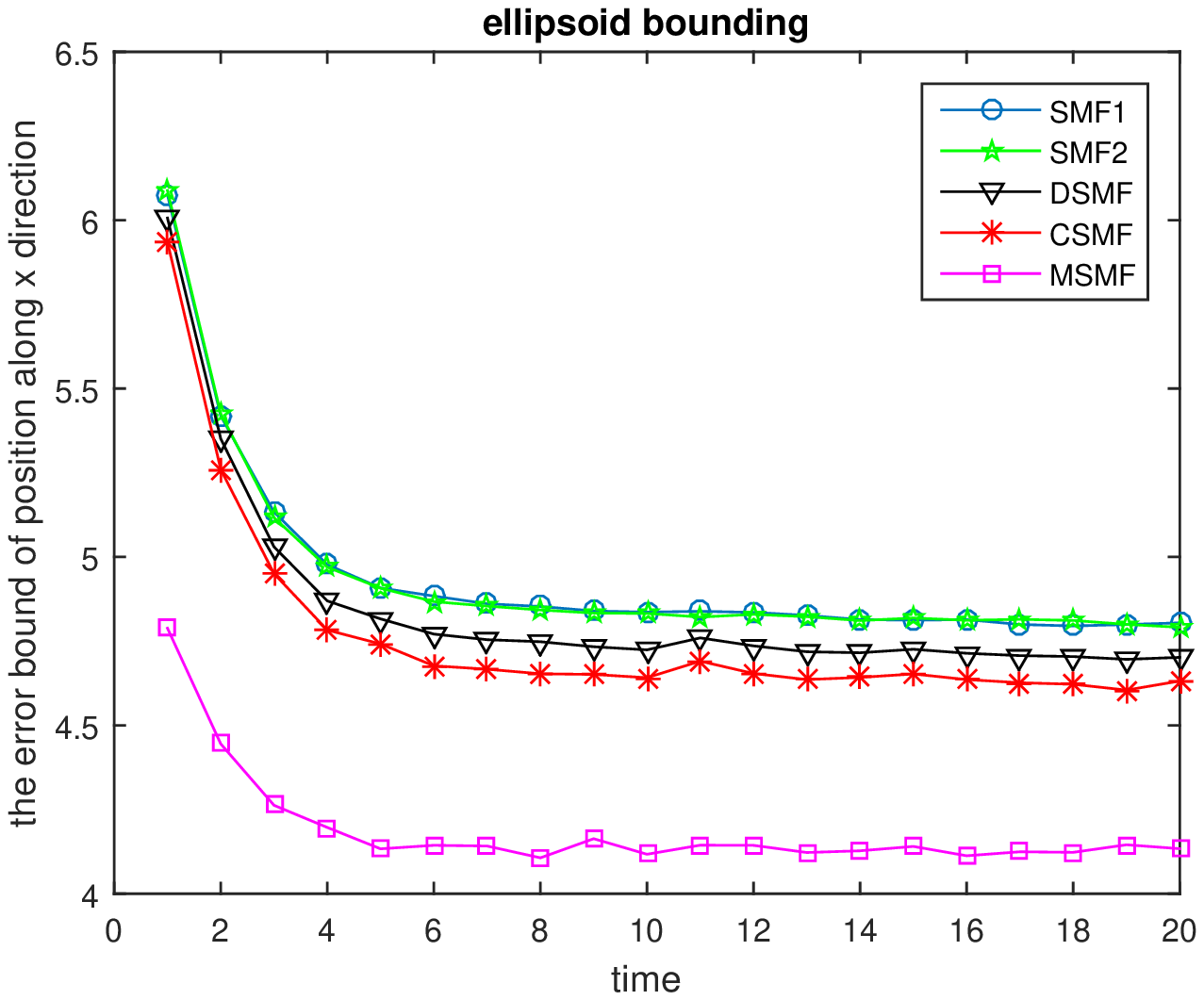}} \hfill}\vfill}
\caption{Comparison of the error bounds of position along $\vx$ direction based on 100
Monte Carlo runs.}\label{fig_01}
\end{figure}

\begin{figure}[h]
\vbox to 6cm{\vfill \hbox to \hsize{\hfill
\scalebox{0.6}[0.59]{\includegraphics{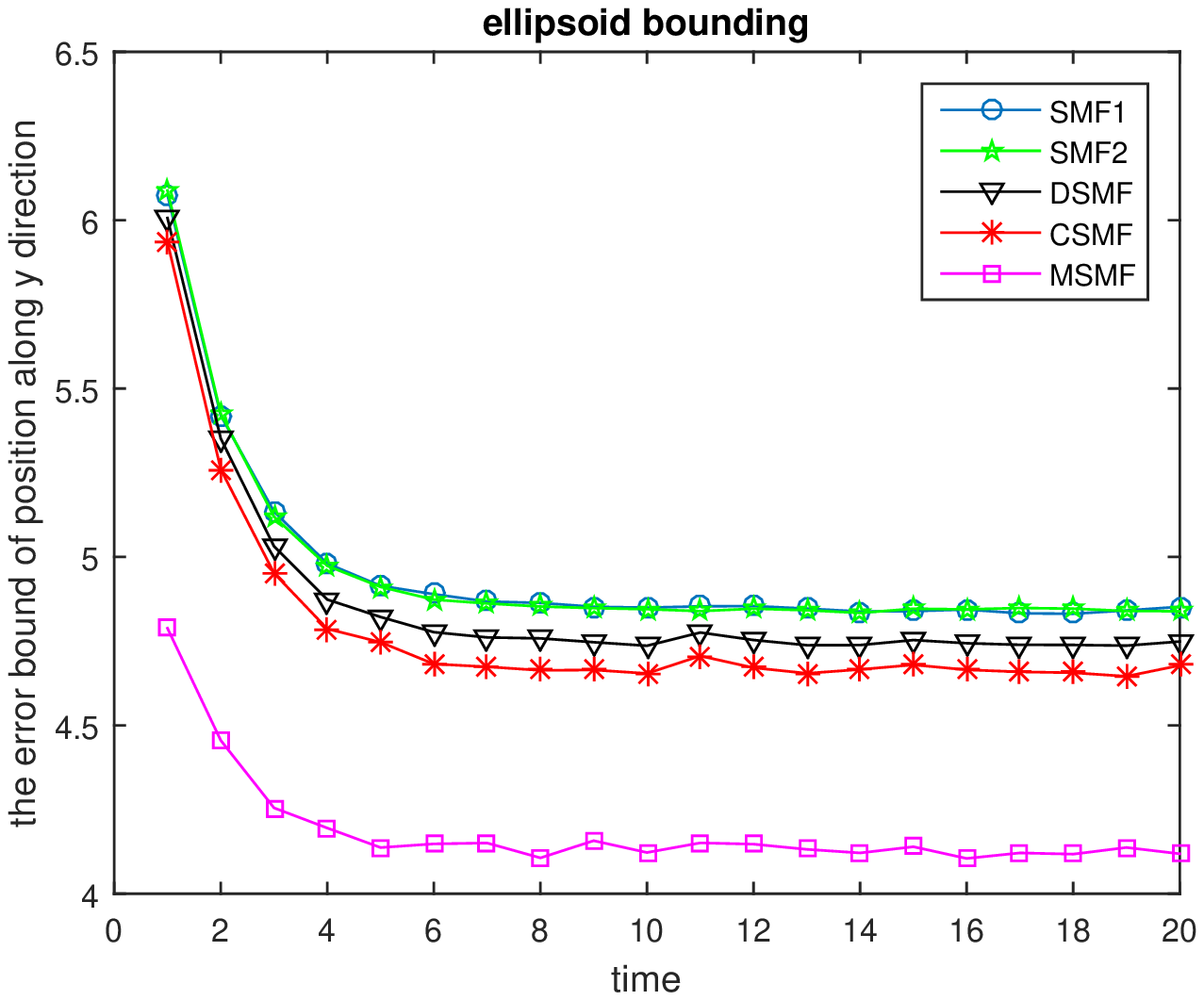}} \hfill}\vfill}
\caption{Comparison of the error bounds of position along $\vy$ direction based on 100
Monte Carlo runs.}\label{fig_02}
\end{figure}

\begin{figure}[h]
\vbox to 6cm{\vfill \hbox to \hsize{\hfill
\scalebox{0.6}[0.59]{\includegraphics{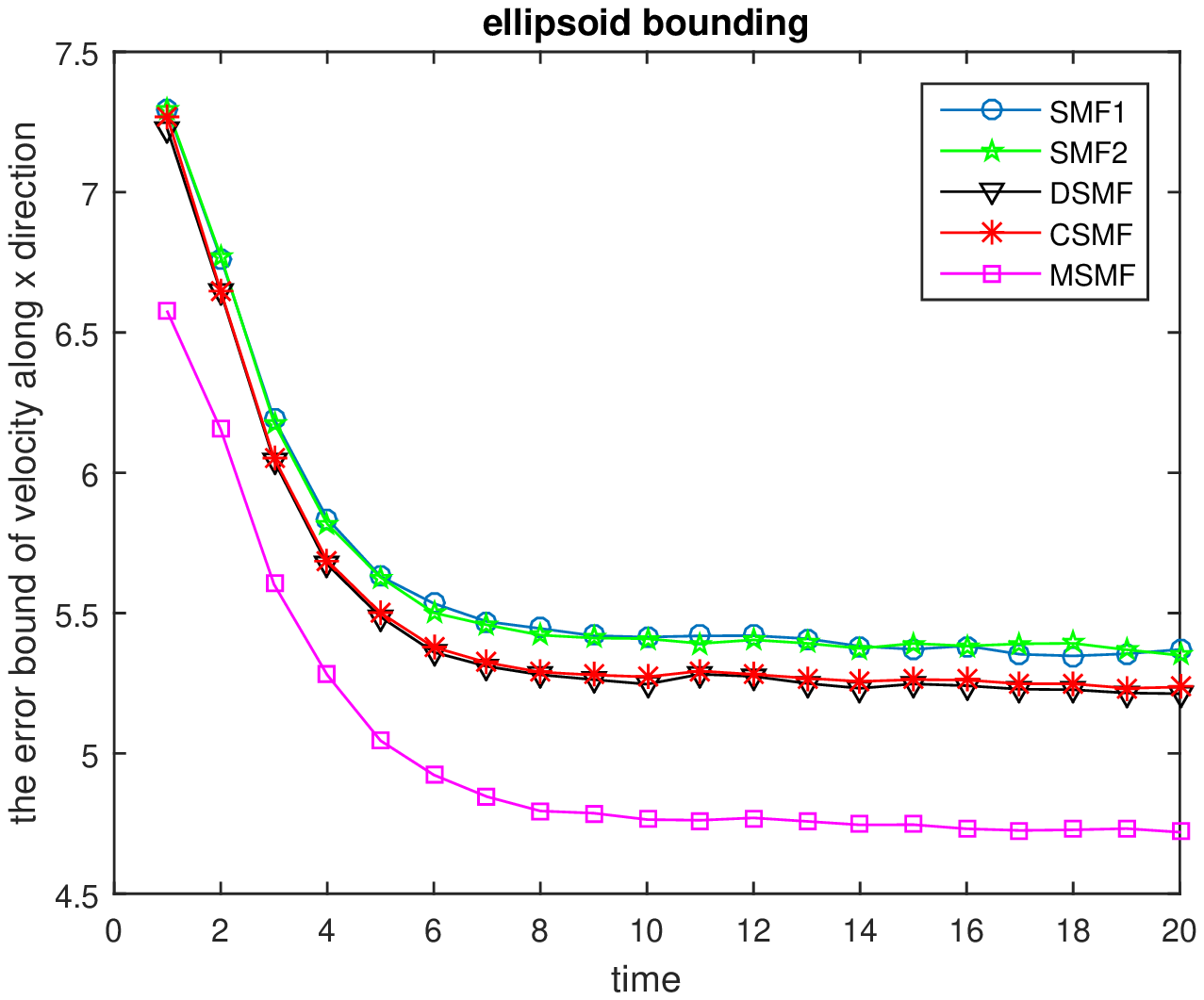}} \hfill}\vfill}
\caption{Comparison of the error bounds of velocity along $\vx$ direction based on 100
Monte Carlo runs.}\label{fig_03}
\end{figure}

\begin{figure}[h]
\vbox to 6cm{\vfill \hbox to \hsize{\hfill
\scalebox{0.6}[0.59]{\includegraphics{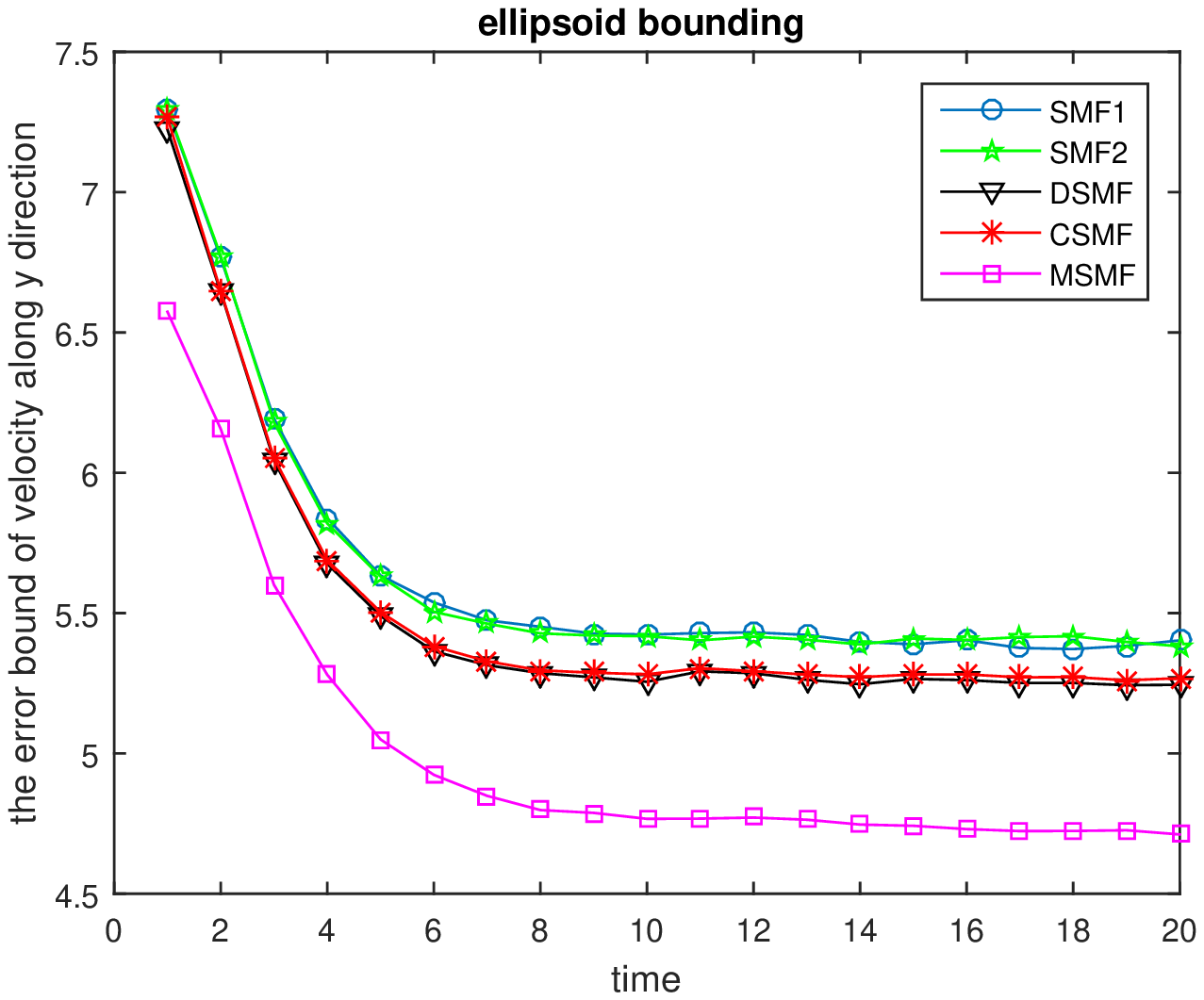}} \hfill}\vfill}
\caption{Comparison of the error bounds of velocity along $\vy$ direction based on 100
Monte Carlo runs.}\label{fig_04}
\end{figure}

\section{Conclusion}\label{sec_6}
This paper has derived the centralized and distributed set-membership information fusion algorithms for multisensor nonlinear dynamic system via minimizing state bounding ellipsoid. Firstly, both of them can be converted into an SDP problem which can be efficiently computed, respectively. Secondly, their analytical solutions can be derived surprisingly by using decoupling technique. It is very interesting that they are quite similar in form to the classic information filter in MSE sense. In the two analytical fusion formulae, the information of each sensor can be clearly characterized, and the knowledge of the correlation among measurement noises across sensors are not required. Finally,  multi-algorithm fusion has been used to minimize the size of the state bounding ellipsoid by complementary advantages of multiple parallel algorithms. A typical example in target tracking has showed that multi-algorithm fusion performs better than both the centralized and distributed fusion. Future work will include, in multisensor nonlinear dynamic system setting, multiple target tracking, sensor management and heterogeneous sensor fusion.

\section{Appendix}\label{sec_7}
\begin{lemma}\label{lem_1}\cite{Boyd-ElGhaoui-Feron-Balakrishnan94}
Let $\mF_0(\eta), \mF_1(\eta),\ldots, \mF_p(\eta)$, be quadratic functions in variable $\eta\in\mathcal {R}^{n}$
\begin{eqnarray}
\mF_i(\eta)=\eta^T\mT_i\eta, ~~i=0,\ldots, p
\end{eqnarray}
with $\mT_i=\mT_i^T$. Then the implication
\begin{eqnarray}
\mF_1(\eta)\leq0,\ldots,\mF_p(\eta)\leq0\Rightarrow\mF_0(\eta)\leq0
\end{eqnarray}
holds if there exist $\tau_1,\ldots,\tau_p\geq0$ such that
\begin{eqnarray}
\mT_0-\sum_{i=1}^p\tau_i\mT_i\preceq0.
\end{eqnarray}
\end{lemma}

\begin{lemma}\label{lem_2}
Schur Complements \cite{Boyd-ElGhaoui-Feron-Balakrishnan94}: Given constant matrices $\mA$, $\mB$, $\mC$, where $\mC=\mC^T$ and $\mA=\mA^T<0$, then
\begin{eqnarray}
\mC-\mB^T\mA^{-1}\mB\preceq0
\end{eqnarray}
if and only if
\begin{eqnarray}
\left[
  \begin{array}{cc}
    \mA & \mB \\
    \mB^T & \mC \\
  \end{array}
\right]\preceq0
\end{eqnarray}
or equivalently
\begin{eqnarray}
\left[
  \begin{array}{cc}
    \mC & \mB^T \\
    \mB & \mA \\
  \end{array}
\right]\preceq0
\end{eqnarray}
\end{lemma}

\begin{lemma}\label{lem_3}
Decoupling \cite{ElGhaoui-Calafiore01}: Let $\mX_{ij}, 1\leq i \leq j \leq 2$ be matrices of appropriate size, with $\mX_{ii}$ square and symmetric. The problem (in variable $\mX,\mZ$)
\begin{eqnarray}\label{Eqpre_30}
\min_{\mX,\mZ}~~f(\mX)~ \mbox{subject to} \left[
                                     \begin{array}{ccc}
                                       \mX & \mZ & \mB \\
                                       \mZ^T & \mX_{11} & \mX_{12} \\
                                       \mB^T & \mX_{12}^T & \mX_{22} \\
                                     \end{array}
                                   \right]\succeq 0
\end{eqnarray}
is feasible if and only if
\begin{eqnarray}\label{Eqpre_31}
                               \left[
                                     \begin{array}{cc}
                                        \mX_{11} & \mX_{12} \\
                                       \mX_{12}^T & \mX_{22} \\
                                     \end{array}
                                   \right]\succeq 0.
\end{eqnarray}
In this case, problem (\ref{Eqpre_30})is equivalent to the problem (in variable $\mX$ only)
\begin{eqnarray}\label{Eqpre_32}
\min_{\mX}~~f(\mX)~ \mbox{subject to} \left[
                                     \begin{array}{cc}
                                        \mX & \mB \\
                                       \mB^T & \mX_{22} \\
                                     \end{array}
                                   \right]\succeq 0.
\end{eqnarray}
Moreover, If the problem (\ref{Eqpre_32}) is feasible, which means that
 \begin{eqnarray}
 \nonumber \mX\succeq\mB\mX_{22}^{+}\mB^T, ~(\mI-\mX_{22}^{+}\mX_{22})\mB^T=0.
\end{eqnarray}
Suppose the objective function is either the trace function or log-det function, then $f(\mX_1)\geq f(\mX_2)$ whenever $\mX_1\succeq\mX_2$. Thus, (\ref{Eqpre_30}) admits a unique optimal variable given by $\mX=\mB\mX_{22}^{+}\mB^T, \mZ=\mB\mX_{22}^{+}\mX_{12}^T$, where $\mX_{22}^{+}$ is the pseudo-inverse of $\mX_{22}$.
\end{lemma}

\begin{proof}[Proof of Lemma \ref{thm_1}]: Note that $\vx_k^c\in\mathcal {E}_{k}^c$ is equivalent to $\vx_k=\hat{\vx}_{k}^c+\mE_{k}^c\vu_{k}$, $\parallel \vu_{k}\parallel\leq 1$,
where $\mE_{k}^c$ is a Cholesky factorization of $\mP_{k}^c$. By the nonlinear state equations (\ref{Eqpre_1}) and (\ref{Eqpre_3}),
\begin{eqnarray}
\nonumber \vx_{k+1}-\hat{\vx}_{k+1|k}^c&=& f_k(\vx_k)+\vw_k-\hat{\vx}_{k+1|k}^c\\
\nonumber  &=&f_k(\hat{\vx}_{k}^c+\mE_{k}^c\vu_{k})+\vw_k-\hat{\vx}_{k+1|k}^c\\
\label{Eqpre_33} &=&f_k(\hat{\vx}_{k}^c)+\mJ_{f_k}\mE_{k}\vu_{k}+\ve_{f_k}
+\mB_{f_k}\Delta_{f_k}+\vw_k-\hat{\vx}_{k+1|k}^c.
\end{eqnarray}
If we denote by
\begin{eqnarray}
\label{Eqpre_34} \xi=[1, ~\vu_{k}^T, ~\vw_k^T, ~\Delta_{f_k}^T]^T,
\end{eqnarray}
then (\ref{Eqpre_33}) can be rewritten as
\begin{eqnarray}
\label{Eqpre_35} \vx_{k+1}-\hat{\vx}_{k+1|k}^c&=&  \Phi_{k+1|k}^c(\hat{\vx}_{k+1|k}^c)\xi
\end{eqnarray}
where $\Phi_{k+1|k}^c(\hat{\vx}_{k+1|k}^c)$ is denoted by (\ref{Eqpre_22}).

Moreover, the condition that $\vx_{k+1}\in\mathcal {E}_{k+1|k}^c $, whenever, I) $\vx_k^c\in\mathcal {E}_k^c$, II) the process noise $\vw_k\in\mW_k$, III) the high-order remainders of state function $\Delta f_k(\vu_k)\in  \mathcal {E}_{f_k} $, which are equivalent to
\begin{eqnarray}
\label{Eqpre_36}  \xi^T\Phi_{k+1|k}(\hat{\vx}_{k+1|k}^c)^T(\mP_{k+1|k}^c)^{-1}\Phi_{k+1|k}(\hat{\vx}_{k+1|k}^c)\xi\leq1,
\end{eqnarray}
whenever
\begin{eqnarray}
\label{Eqpre_37}  \parallel \vu_{k}\parallel&\leq& 1 ,\\
\label{Eqpre_38}  \vw_k^T\mQ_k^{-1}\vw_k&\leq& 1,\\
\label{Eqpre_39}\parallel \Delta_{f_k}\parallel&\leq &1.
\end{eqnarray}
The equations (\ref{Eqpre_37})--(\ref{Eqpre_39}) are equivalent to
\begin{eqnarray}
\label{Eqpre_40}  \xi^T\diag(-1,\mI,0,0)\xi&\leq& 0,\\
\label{Eqpre_41}  \xi^T\diag(-1,0,\mQ_k^{-1},0)\xi&\leq& 0,\\
\label{Eqpre_42}  \xi^T\diag(-1,0,0,\mI)\xi&\leq& 0.
\end{eqnarray}
where $\mI$ and $0$ are matrices with compatible dimensions.

From Lemma \ref{lem_1}, a sufficient condition such that the inequalities
(\ref{Eqpre_40})-(\ref{Eqpre_42}) imply (\ref{Eqpre_36}) to hold is that
there exist nonnegative scalars
$\tau^u\geq0, \tau^w\geq0,\tau^f\geq0$, such that
\begin{eqnarray}
\nonumber&& \Phi_{k+1|k}(\hat{\vx}_{k+1|k}^c)^T(\mP_{k+1|k}^c)^{-1}\Phi_{k+1|k}(\hat{\vx}_{k+1|k}^c)\\
\nonumber&&-\diag(1,0,0,0)\\
\nonumber&&-\tau^u\diag(-1,I,0,0)\\
\nonumber&&-\tau^w\diag(-1,0,\mQ_k^{-1},0)\\
\label{Eqpre_43}&&-\tau^f\diag(-1,0,0,I)\preceq0
\end{eqnarray}
Furthermore, (\ref{Eqpre_43}) is
written in the following compact form:
\begin{eqnarray}
\label{Eqpre_44} \Phi_{k+1|k}(\hat{\vx}_{k+1|k}^c)^T(\mP_{k+1|k}^c)^{-1}\Phi_{k+1|k}(\hat{\vx}_{k+1|k}^c)-\Xi\preceq0
\end{eqnarray}
where $\Xi$ is denoted by (\ref{Eqpre_23}). Applying Lemma \ref{lem_2}, (\ref{Eqpre_44}) is
equivalent to
\begin{eqnarray}
\label{Eqpre_45}&&\left[\begin{array}{cc}
    \mP_{k+1|k}^c&\Phi_{k+1|k}(\hat{\vx}_{k+1|k}^c)\\
    (\Phi_{k+1|k}(\hat{\vx}_{k+1|k}^c))^T& ~~\Xi\\
\end{array}\right]\succeq0\\
\label{Eqpre_46}&&\mP_{k+1|k}^c\succ0.
\end{eqnarray}
Therefore, if
$\hat{x}_{k+1|k}^c$, $\mP_{k+1|k}^c$ satisfy (\ref{Eqpre_45}), then the state $x_{k+1}$ belongs to $\mathcal
{E}_{k+1|k}^c$, whenever, I) $\vx_k^c$ is in $\mathcal {E}_k^c$, II) the process noise $\vw_k\in\mW_k$, III) the high-order remainders of state function $\Delta f_k(\vu_k)\in  \mathcal {E}_{f_k} $.

Summarizing  the above results,  the computation of the
predicted  bounding ellipsoid by minimizing a
size measure $f(\mP_{k+1|k}^c)$ (\ref{Eqpre_18}) is Lemma \ref{thm_1}.
\end{proof}
\begin{proof}[Proof of Theorem \ref{cor_1}]: If we partition the left side of (\ref{Eqpre_21}) by appropriate block, then it can be rewritten as
\begin{eqnarray}\label{Eqpre_29}
                                  \left[
                                     \begin{array}{ccc}
                                       \mP_{k+1|k}^c & \mZ & \mB \\
                                       \mZ^T & \mX_{11} & \mX_{12} \\
                                       \mB^T & \mX_{12}^T & \mX_{22} \\
                                     \end{array}
                                   \right]\succeq 0,
\end{eqnarray}
where
\begin{eqnarray}
\nonumber\mZ&=&f_k(\hat{\vx}_{k}^c)+\ve_{f_k}-\hat{\vx}_{k+1|k},\\
\nonumber\mB&=&[\mJ_{f_k}\mE_{k}^c, ~\mI,~\mB_{f_k}],\\
\nonumber\mX_{11}&=&1-\tau^u-\tau^w-\tau^f,\\
\nonumber\mX_{22}&=&\diag(\tau^u\mI,\tau^w\mQ_k^{-1},\tau^f\mI),\\
\nonumber\mX_{12}&=&0.
\end{eqnarray}
Based on the decoupling technique in Lemma \ref{lem_3}, the above matrix inequality is feasible if and only if
\begin{eqnarray}
 \nonumber                                 \left[
                                     \begin{array}{cc}
                                       \mX_{11} & \mX_{12} \\
                                        \mX_{12}^T & \mX_{22} \\
                                     \end{array}
                                   \right]\succeq 0.
\end{eqnarray}
From the expression of $\mX_{11}, \mX_{12}, \mX_{22}$, it is also equivalent to
\begin{eqnarray}
\nonumber \tau^u+\tau^w+\tau^f\leq1, ~-\tau^u\leq0,~ -\tau^w\leq0, ~-\tau^f\leq0.
\end{eqnarray}
Thus, the optimization problem of Lemma \ref{thm_1}
\begin{eqnarray}
 \nonumber \min_{\tau^u,\tau^w,\tau^f}\min_{\mP_{k+1|k}^c,\hat{\vx}_{k+1|k}^c}~~ tr(\mP_{k+1|k}^c)~\mbox{subject to}~ (\ref{Eqpre_19})-(\ref{Eqpre_20})~ and ~(\ref{Eqpre_29}),
\end{eqnarray}
which, by Lemma \ref{lem_3}, is equivalent to
\begin{eqnarray}
 \label{Eqpre_47} &&\min~~ tr(\mB\mX_{22}^{+}\mB^T)\\
  \nonumber&&~~\mbox{subject to} -\tau^u\leq0,~ -\tau^w\leq0,~ -\tau^f\leq0,~\tau^u+\tau^w+\tau^f\leq1\\
 \label{Eqpre_48} &&~~\qquad\qquad(\mI-\mX_{22}^{+}\mX_{22})\mB^T=0.
\end{eqnarray}
It is easy to see that $\mX_{22}$ is nonsingular according to (\ref{Eqpre_48}), then, the above optimization problem is equivalent to
\begin{eqnarray}
\nonumber &&\min~~ tr(\frac{\mJ_{f_k}\mP_k^c\mJ_{f_k}^T}{\tau^u}+\frac{\mQ_k}{\tau^w}+\frac{\mP_{f_k}}{\tau^f})\\
  \nonumber&&~~\mbox{subject to} -\tau^u<0,~ -\tau^w<0,~ -\tau^f<0,~\tau^u+\tau^w+\tau^f\leq1
\end{eqnarray}
where $\mP_{f_k}=\mB_{f_k}\mB_{f_k}^T$.
 Therefore, based on Lagrange dual function, the analytically optimal solution can be obtained in (\ref{Eqpre_24})-(\ref{Eqpre_28}).
\end{proof}

\begin{proof}[Proof of Lemma \ref{thm_2}]: Note that we have get $\vx_{k+1}\in\mathcal {E}_{k+1|k}^c$ in prediction step, which is equivalent to $\vx_{k+1}=\hat{\vx}_{k+1|k}^c+\mE_{k+1|k}^c\vu_{k+1|k}$, $\parallel\vu_{k+1|k}\parallel\leq 1$,
where $\mE_{k+1|k}^c$ is a Cholesky factorization of $\mP_{k+1|k}^c$, then,
\begin{eqnarray}
\label{Eqpre_63}  \vx_{k+1}-\hat{\vx}_{k+1}^c&=& \hat{\vx}_{k+1|k}^c+\mE_{k+1|k}^c\vu_{k+1|k}-\hat{\vx}_{k+1}^c
\end{eqnarray}
and by the nonlinear measurement equations (\ref{Eqpre_2}) and (\ref{Eqpre_4})
\begin{eqnarray}
\nonumber \vy_{k+1}^i&=& h_{k+1}^i(\vx_{k+1})+\vv_{k+1}^i\\
\label{Eqpre_64}   &=&h_{k+1}^i(\hat{\vx}_{k+1|k}^c)+\mJ_{h_{k+1}^i}\mE_{k+1|k}\vu_{k+1|k}
+\ve_{h_{k+1}^i}+\mB_{h_{k+1}^i}\Delta_{h_{k+1}^i}+\vv_{k+1}^i
\end{eqnarray}
If we denote by
\begin{eqnarray}
\label{Eqpre_65} \xi=[1, ~\vu_{k+1|k}^T, ~\underbrace{\vv_{k+1}^{1^T},\ldots, ~\vv_{k+1}^{L^T}}_{L~blocks}, ~\underbrace{\Delta_{h_{k+1}^1}^{T},\ldots,~\Delta_{h_{k+1}^L}^{T}}_{L~blocks}]^T,
\end{eqnarray}
then (\ref{Eqpre_63}) and (\ref{Eqpre_64}) can be rewritten as
\begin{eqnarray}
\label{Eqpre_66} \vx_{k+1}-\hat{\vx}_{k+1}^c&=&  \Phi_{k+1}^c(\hat{\vx}_{k+1}^c)\xi\\
\label{Eqpre_67} 0&=&  \Psi_{k+1}^c(\vy_{k+1}^i)\xi,
\end{eqnarray}
where $\Phi_{k+1}^c(\hat{\vx}_{k+1}^c)$ and $\Psi_{k+1}^c(\vy_{k+1}^i)$ are denoted by (\ref{Eqpre_53}) and (\ref{Eqpre_54}), respectively.

Moreover, the condition that $\vx_{k+1}\in\mathcal {E}_{k+1}^c$ whenever I) $\vx_{k+1}$ is in $\mathcal {E}_{k+1|k}^c$
II) measurement noises $\vv_{k+1}^i$  are bounded
in ellipsoidal sets, i.e., $\vv_{k+1}^i\in\mV_{k+1}^i$,  III) the high-order remainders of measurement function $\Delta_{h_{k+1}^i}\in  \mathcal {E}_{h_{k+1}^i} $, , $i=1,\ldots,L$, which are equivalent to
\begin{eqnarray}
\label{Eqpre_68}  \xi^T\Phi_{k+1}^c(\hat{\vx}_{k+1}^c)^T(\mP_{k+1}^c)^{-1}\Phi_{k+1}^c(\hat{\vx}_{k+1}^c)\xi\leq1,
\end{eqnarray}
whenever
\begin{eqnarray}
\label{Eqpre_69}  \parallel \vu_{k+1|k}\parallel&\leq& 1 ,\\
\label{Eqpre_70} \vv_{k+1}^{i^T}\mR_{k+1}^{i^{-1}}\vv_{k+1}^i&\leq& 1,\\
\label{Eqpre_71}\parallel \Delta_{h_{k+1}^i}\parallel&\leq &1,~~i=1,\ldots,L.
\end{eqnarray}
The equations (\ref{Eqpre_69})--(\ref{Eqpre_71}) are equivalent to
\begin{eqnarray}
\label{Eqpre_72}  \xi^T\diag(-1,\mI,\underbrace{0,\ldots,0}_{L~blocks},\underbrace{0,\ldots,0}_{L~blocks})\xi&\leq& 0,\\
\label{Eqpre_73}  \xi^T\diag(-1,0,\vdots\underbrace{0,\ldots,\mR_{k+1}^{i^{-1}},\ldots,0}_{the~i-th~ block~ is~ \mR_{k+1}^{i^{-1}} }\vdots,\underbrace{0,\ldots,0}_{L~blocks})\xi&\leq& 0,\\
\label{Eqpre_74}  \xi^T\diag(-1,0,\underbrace{0,\ldots,0}_{L~blocks},\vdots\underbrace{0,\ldots,\mI,\ldots,0}_{the~i-th~ block~ is~ \mI}\vdots)\xi&\leq& 0,
\end{eqnarray}
where $\mI$ and $0$ are matrices with compatible dimensions.

By  $\mathcal {S}$-procedure Lemma \ref{lem_1} and
(\ref{Eqpre_67}), a sufficient condition such that the inequalities
(\ref{Eqpre_72})-(\ref{Eqpre_74}) imply (\ref{Eqpre_68}) to hold is that
there exist scalars $\tau_i^y$ and nonnegative scalars
$\tau^u\geq0, \tau_i^v\geq0, \tau_i^h\geq0$, such that
\begin{eqnarray}
\nonumber&& \Phi_{k+1}^c(\hat{\vx}_{k+1}^c)^T(\mP_{k+1}^c)^{-1}\Phi_{k+1}^c(\hat{\vx}_{k+1}^c)\\
\nonumber&&-\diag(1,0,\vdots0,\ldots,0,\vdots0,\ldots,0,\vdots0,\ldots,0)\\
\nonumber&&-\tau^u\diag(-1,\mI,\vdots0,\ldots,0,\vdots0,\ldots,0,\vdots0,\ldots,0)\\
\nonumber&&-\sum_{i=1}^L\tau_i^v\diag(-1,0,\vdots\underbrace{0,\ldots,\mR_{k+1}^{i^{-1}},\ldots,0,}_{the~i-th~ block~ is~ \mR_{k+1}^{i^{-1}} }\vdots0,\ldots,0,\vdots0,\ldots,0)\\
\nonumber&&-\sum_{i=1}^L\tau_i^h\diag(-1,0,\vdots0,\ldots,0,\vdots0,\ldots,0,\vdots\underbrace{0,\ldots,\mI,\ldots,0}_{the~i-th~ block~ is~ \mI})\\
\label{Eqpre_75}&&-\sum_{i=1}^L\tau_i^y \Psi_{k+1}^c(\vy_{k+1}^i)^T \Psi_{k+1}^c(\vy_{k+1}^i) \preceq0
\end{eqnarray}
Furthermore, (\ref{Eqpre_75}) is
written in the following compact form:
\begin{eqnarray}
\label{Eqpre_76} \Phi_{k+1}^c(\hat{\vx}_{k+1}^c)^T(\mP_{k+1}^c)^{-1}\Phi_{k+1}^c(\hat{\vx}_{k+1}^c)-\Xi- (\Psi_{k+1}^c)^T \diag(\tau_1^y,\ldots,\tau_L^y)\Psi_{k+1}^c \preceq0
\end{eqnarray}
where $\Xi$ and $\Psi_{k+1}^c$ are denoted by (\ref{Eqpre_56}) and (\ref{Eqpre_55}), respectively.

If we denote
$(\Psi_{k+1}^c)_{\bot}$ is the orthogonal complement of $\Psi_{k+1}^c$, then
(\ref{Eqpre_76}) is equivalent to
\begin{eqnarray}
\nonumber && ((\Psi_{k+1}^c)_{\bot})^T\Phi_{k+1}^c(\hat{\vx}_{k+1}^c)^T(\mP_{k+1}^c)^{-1}\Phi_{k+1}^c(\hat{\vx}_{k+1}^c)(\Psi_{k+1}^c)_{\bot}\\
\label{Eqpre_77}&&-((\Psi_{k+1}^c)_{\bot})^T\Xi(\Psi_{k+1}^c)_{\bot} \preceq0
\end{eqnarray}
Using Schur complements Lemma \ref{lem_2}, (\ref{Eqpre_77}) is
equivalent to
\begin{eqnarray}
\label{Eqpre_78}&&\left[\begin{array}{cc}
    -\mP_{k+1}^c&\Phi_{k+1}^c(\hat{\vx}_{k+1}^c)(\Psi_{k+1}^c)_{\bot}\\
    (\Phi_{k+1}^c(\hat{\vx}_{k+1}^c)(\Psi_{k+1}^c)_{\bot})^T& ~~-(\Psi_{k+1}^c)_{\bot}^T\Xi(\Psi_{k+1}^c)_{\bot}\\
\end{array}\right]\preceq0.\\
\label{Eqpre_79}&&-\mP_{k+1}^c\prec0.
\end{eqnarray}

Therefore, if
$\hat{x}_{k+1}^c$, $\mP_{k+1}^c$ satisfy (\ref{Eqpre_78})-(\ref{Eqpre_79}), then the state $x_{k+1}$ belongs to $\mathcal
{E}_{k+1}^c$, whenever I) $\vx_{k+1}$ is in $\mathcal {E}_{k+1|k}^c$
II) measurement noises $\vv_{k+1}^i$ are bounded
in ellipsoidal sets, i.e., $\vv_{k+1}^i\in\mV_{k+1}^i$, III) the high-order remainders of measurement function $\Delta_{h_{k+1}^i}\in  \mathcal {E}_{h_{k+1}^i} $, , $i=1,\ldots,L$.

Summarizing  the above results,  the computation of the
measurement update  bounding ellipsoid by minimizing a
size measure $f(\mP_{k+1}^c)$ (\ref{Eqpre_49}) is  Lemma \ref{thm_2}.
\end{proof}

\begin{proof}[Proof of Theorem \ref{cor_2}]:
In view of the optimization problem in Lemma \ref{thm_2}, we can apply Lemma \ref{lem_3} to the linear matrix inequalities (\ref{Eqpre_80}), with $\mZ=\hat{\vx}_{k+1}^c-\hat{\vx}_{k+1|k}^c$, and the rest of matrices defined appropriately. Thus, the problem
\begin{eqnarray}
 \nonumber \min_{\tau^u,\tau_i^v,\tau_i^h}\min_{\mP_{k+1}^c,\hat{\vx}_{k+1}^c}~~ f(\mP_{k+1}^c)~\mbox{subject to}~ (\ref{Eqpre_50}),~(\ref{Eqpre_51})~ and ~(\ref{Eqpre_80}),
\end{eqnarray}
which is equivalent to
\begin{eqnarray}
 \nonumber \min_{\tau^u,\tau_i^v,\tau_i^h}~~ f(\bar{\mX}(\tau^u,\tau_i^v,\tau_i^h))~\mbox{subject to}~ (\ref{Eqpre_58}), (\ref{Eqpre_59}),(\mI-(\Psi_{22}^T\Xi_{22}\Psi_{22})^{+}\Psi_{22}^T\Xi_{22}\Psi_{22})\mB^T=0,
\end{eqnarray}
where $\bar{\mX}(\tau^u,\tau_i^v,\tau_i^h)=\mB(\Psi_{22}^T\Xi_{22}\Psi_{22})^{+}\mB^T$, $i=1,\ldots,L$.

If one of $\tau^u,\tau_i^v,\tau_i^h$, $i=1,\ldots,L$, is zero, then the feasible sets of $\mP_{k+1}^c$ and $\hat{\vx}_{k+1}^c$ become smaller from (\ref{Eqpre_76}), and the  objective value  becomes larger. Thus, the optimal $\tau^u,\tau_i^v,\tau_i^h$, $i=1,\ldots,L$ should be greater than zero, and $\Psi_{22}^T\Xi_{22_{opt}}\Psi_{22}$ be nonsingular.
If $\mB(\Psi_{22}^T\Xi_{22_{opt}}\Psi_{22})^{-1}\mB^T$ is the optimal value of the above optimization problem, then, by using Lemma \ref{lem_3} again, the optimal ellipsoid $\mathcal{E}_{k+1}^c$ is given by
\begin{eqnarray}
\label{Eqpre_200} \mP_{k+1}^c&=&\mB(\Psi_{22}^T\Xi_{22_{opt}}\Psi_{22})^{-1}\mB^T,\\
\label{Eqpre_201} \mZ&=&\mB(\Psi_{22}^T\Xi_{22_{opt}}\Psi_{22})^{-1} \Psi_{22}^T\Xi_{22_{opt}}\Psi_{21}.
\end{eqnarray}
Based on (\ref{Eqpre_201}) and $\mZ=\hat{\vx}_{k+1}^c-\hat{\vx}_{k+1|k}^c$, we retrieve the center of the ellipsoid as
\begin{eqnarray}
\label{Eqpre_202}\hat{\vx}_{k+1}^c=\hat{\vx}_{k+1|k}^c+\mB(\Psi_{22}^T\Xi_{22_{opt}}\Psi_{22})^{-1}(\Psi_{22}^T\Xi_{22_{opt}}\Psi_{21}).
\end{eqnarray}

By the definition of $\Psi_{22}$ and $\Xi_{22}$ in (\ref{Eqpre_118}) and  (\ref{Eqpre_120}),
\begin{eqnarray}
\nonumber &&\Psi_{22}^T\Xi_{22}\Psi_{22}=\\
\nonumber
                              &&\left[
                               \begin{array}{cccc}
                                 \tau^u\mP_{k+1|k}^{c^{-1}}+\sum_{i=1}^L\tau_i^v{\mJ_{h_{k+1|k}^i}^T}\mR_{k+1}^{i^{-1}}\mJ_{h_{k+1|k}^i} & -\tau_1^v{\mJ_{h_{k+1|k}^1}^T}\mR_{k+1}^{1^{-1}} & \ldots &  -\tau_L^v{\mJ_{h_{k+1|k}^L}^T}\mR_{k+1}^{L^{-1}} \\
                                 -\tau_1^v({\mJ_{h_{k+1|k}^1}^T}\mR_{k+1}^{1^{-1}})^T& \tau_1^v\mR_{k+1}^{1^{-1}}+\tau_1^h\mP_{h_{k+1}^1}^{-1} &\ldots & 0 \\
                                 \vdots & \vdots & \ddots & \vdots \\
                                 -\tau_L^v({\mJ_{h_{k+1|k}^L}^T}\mR_{k+1}^{L^{-1}})^T & 0 & \ldots & \tau_L^v\mR_{k+1}^{L^{-1}}+\tau_L^h\mP_{h_{k+1}^L}^{-1} \\
                               \end{array}
                             \right]
\end{eqnarray}
then
\begin{eqnarray}
\nonumber &&\mB(\Psi_{22}^T\Xi_{22}\Psi_{22})^{-1}\mB^T=[\mI ~\underbrace{0,\ldots, 0}_{L~ blocks}](\Psi_{22}^T\Xi_{22}\Psi_{22})^{-1}[\mI ~\underbrace{0,\ldots, 0}_{L~ blocks}]^T\\
\nonumber &=&\left(\tau^u\mP_{k+1|k}^{c^{-1}}+\sum_{i=1}^L{\mJ_{h_{k+1|k}^i}^T}(\frac{\mR_{k+1}^{i}}{\tau_i^v}+\frac{\mP_{h_{k+1}^i}}{\tau_i^h})^{-1}
\mJ_{h_{k+1|k}^i}\right)^{-1}.
\end{eqnarray}
Thus, (\ref{Eqpre_62}) can be obtained by (\ref{Eqpre_200}). Moreover, substituting (\ref{Eqpre_122}), (\ref{Eqpre_118}) and  (\ref{Eqpre_120}) into (\ref{Eqpre_202}), then (\ref{Eqpre_180}) can be achieved.
\end{proof}

\begin{proof}[Proof of Lemma \ref{thm_3}]: Note that $\vx_{k+1}\in\mathcal {E}_{k+1|k}^d$ is equivalent to $\vx_{k+1}=\hat{\vx}_{k+1|k}^d+\mE_{k+1|k}^d\vu_{k+1|k}$, $\parallel \vu_{k+1|k}\parallel\leq 1$,
where $\mE_{k+1|k}^d$ is a Cholesky factorization of $\mP_{k+1}^d$, then
\begin{eqnarray}
\label{Eqpre_89} \vx_{k+1}-\hat{\vx}_{k+1}^d&=& \hat{\vx}_{k+1|k}^d+\mE_{k+1|k}^d\vu_{k+1|k}-\hat{\vx}_{k+1}^d.
\end{eqnarray}
If we denote by
\begin{eqnarray}
\label{Eqpre_90} \xi=[1, ~\vu_{k+1|k}^T]^T,
\end{eqnarray}
then (\ref{Eqpre_89}) can be rewritten as
\begin{eqnarray}
\label{Eqpre_91} \vx_{k+1}-\hat{\vx}_{k+1}^d&=& \Phi_{k+1}^d\xi
\end{eqnarray}
where $\Phi_{k+1}^d$ is denoted by (\ref{Eqpre_85}). Similarly, we have
\begin{eqnarray}
\label{Eqpre_92} \vx_{k+1}-\hat{\vx}_{k+1}^i&=&  \Phi_{k+1}^i\xi
\end{eqnarray}
where $\Phi_{k+1}^i$ is denoted by (\ref{Eqpre_86}).

Moreover, the condition that $\vx_{k+1}\in\mathcal {E}_{k+1}^d$, whenever, I) $\vx_{k+1}$ is in $\mathcal {E}_{k+1|k}^d$, II)$\vx_{k+1}\in\mathcal {E}_{k+1}^i$, for $i=1,\ldots,L$, is equivalent to
\begin{eqnarray}
\label{Eqpre_93}  \xi^T(\Phi_{k+1}^d)^T(\mP_{k+1}^d)^{-1}\Phi_{k+1}^d\xi\leq1,
\end{eqnarray}
whenever, for $i=1,\ldots,L$,
\begin{eqnarray}
\label{Eqpre_94}  \parallel \vu_{k+1|k}\parallel&\leq& 1 ,\\
\label{Eqpre_97} \xi^T(\Phi_{k+1}^i)^T(\mP_{k+1}^i)^{-1}\Phi_{k+1}^i\xi&\leq &1,
\end{eqnarray}
The equations (\ref{Eqpre_94})--(\ref{Eqpre_97}) are equivalent to
\begin{eqnarray}
\label{Eqpre_98}  \xi^T\diag(-1,I)\xi&\leq& 0,\\
\label{Eqpre_101}   \xi^T[(\Phi_{k+1}^i)^T(\mP_{k+1}^i)^{-1}\Phi_{k+1}^i+\diag(-1,0)]\xi&\leq& 0,
\end{eqnarray}
where $I$ and $0$ are matrices with compatible dimensions.

By  $\mathcal {S}$-procedure Lemma \ref{lem_1}, a sufficient condition such that the inequalities
(\ref{Eqpre_98})-(\ref{Eqpre_101}) imply (\ref{Eqpre_93}) to hold is that
there exist nonnegative scalars
$\tau^u\geq0, \tau_i^y\geq0$, $i=1,\ldots,L$, such that
\begin{eqnarray}
\nonumber&& (\Phi_{k+1}^d)^T(\mP_{k+1}^d)^{-1}\Phi_{k+1}^d-\diag(1,0)-\tau^u\diag(-1,I)\\
\label{Eqpre_102}&&~~-\sum_{i=1}^L\tau_i^y [(\Phi_{k+1}^i)^T (\mP_{k+1}^i)^{-1}\Phi_{k+1}^i+\diag(-1,0)] \preceq0
\end{eqnarray}
Furthermore, (\ref{Eqpre_102}) is
written in the following compact form:
\begin{eqnarray}
\label{Eqpre_103} (\Phi_{k+1}^d)^T(\mP_{k+1}^d)^{-1}\Phi_{k+1}^d-\Xi- \Pi \preceq0
\end{eqnarray}
where $\Xi$ and $\Pi$ are denoted by (\ref{Eqpre_87}) and (\ref{Eqpre_88}), respectively.

Using Schur complements Lemma \ref{lem_2}, (\ref{Eqpre_103}) is
equivalent to
\begin{eqnarray}
\label{Eqpre_104}&&\left[\begin{array}{cc}
    -\mP_{k+1}^d&\Phi_{k+1}^d\\
    (\Phi_{k+1}^d)^T& ~~-\Xi- \Pi\\
\end{array}\right]\preceq0\\
\label{Eqpre_105}&&-\mP_{k+1}^d\prec0.
\end{eqnarray}

Therefore, if
$\hat{x}_{k+1|k}^d$, $\mP_{k+1}^d$ satisfy (\ref{Eqpre_104})-(\ref{Eqpre_105}), then the state $x_{k+1}$ belongs to $\mathcal
{E}_{k+1}^d$, whenever, I) $\vx_{k+1}$ is in $\mathcal {E}_{k+1|k}^d$, II) $x_{k+1}$ belongs to $\mathcal{E}_{k+1}^i$, for $i=1,\ldots,L$.

Summarizing  the above results,  the computation of the
 bounding ellipsoid for distributed fusion by minimizing a
size measure $f(\mP_{k+1}^d)$ (\ref{Eqpre_81}) is  Lemma \ref{thm_3}.
\end{proof}

\end{document}